\documentclass{article}

\usepackage[a4paper, margin=1in]{geometry}
\usepackage{graphicx}
\usepackage{enumerate}
\usepackage{natbib}
\usepackage{amsmath}
\usepackage{amssymb}
\usepackage{amsthm}

\usepackage{mathtools}
\usepackage{multirow}
\usepackage{pdflscape}
\usepackage{bm}
\usepackage{bbm}
\usepackage{float}
\usepackage{longtable}

\usepackage{appendix}
\usepackage{url}

\newtheorem{theorem}{Theorem}

\newtheorem{corollary}{Corollary}

\newtheorem{example}{Example}
\newtheorem{remark}{Remark}
\newtheorem{definition}{Definition}

\title{Sequential changepoint detection in classification data
under label shift}

\usepackage{authblk} % author list

\author[1]{Ciaran Evans}
\author[2]{Max G'Sell}
\affil[1]{Department of Mathematics and Statistics, Wake Forest University}
\affil[2]{Department of Statistics and Data Science, Carnegie Mellon University}
\begin{document}

\maketitle

\begin{abstract}
Classifier predictions often rely on the assumption that new
  observations come from the same distribution as training data. When the
  underlying distribution changes, so does the optimal classification rule, and
  performance may degrade. We consider the problem of detecting such
  a change in distribution in sequentially-observed, unlabeled classification
  data.
  We focus on label shift changes to the distribution, where the class priors
  shift but the class conditional distributions remain unchanged.  
  %We consider the problem of detecting a
  %change to the overall fraction of positive cases, known as label shift, in
  %sequentially-observed binary classification data. 
  We reduce this problem to
  the problem of detecting a change in the one-dimensional classifier scores,
  leading to simple nonparametric 
  sequential changepoint detection procedures. Our procedures leverage
  classifier training data to estimate the detection statistic, and converge to
  their parametric counterparts in the size of the training data. In
  simulations, we show that our method outperforms other detection
  procedures in this label shift setting.
\end{abstract}

\section{Introduction}

We consider the problem of rapid, online detection of a change in the
distribution of classification data, without access to the true classification
labels of those data points.  This problem is of importance both because such a
change impacts the performance of classification algorithms, and because
it often reflects an interesting change in the underlying generating process.  The
non-sequential problem of classification under a changed generating distribution has been
extensively studied, as has the general problem of sequential
changepoint detection; see Section
\ref{sec:background}.  However, the intersection of these topics is
relatively unexplored, and yields interesting structure and 
methodological improvements.  In this paper, we restrict ourselves to the
particular case of a \emph{label shift} change in the
distribution \citep{lipton2018detecting}, where the distribution of
classification labels changes without changing the conditional distribution of
the covariates; we illustrate this setting in the following motivating example.

%\vspace{-0.25em}
\subsection{Motivation: dengue outbreaks} 

Dengue, a viral infection transmitted by
mosquitoes, is found in tropical and sub-tropical regions around the world, and
affects up to 400 million people a year \citep{world2020dengue}. Early
treatment of dengue is important for improving prognosis, and so it is key to
correctly diagnose patients with the disease. However, dengue cases are
commonly mis-diagnosed \citep{world2020dengue}; while gold-standard diagnostic
tests and rapid antigen tests exist, these may not always be available to
healthcare providers. To assist healthcare workers in diagnosis and early
detection of dengue, \cite{tuan2015sensitivity} developed a classifier based on
simple diagnostic and laboratory measurements, such as temperature, vomiting,
and white blood cell count. The authors recommend deploying the classifier to
help diagnose dengue in patients, which entails sequentially applying the
classifier to make a prediction for each new patient.

However, the prevalence of dengue in a community may change quickly, due to
both seasonal trends and outbreaks \citep{wiwanitkit2006observation,
garg2011prevalence, hsu2017trend}. When a sudden change in dengue prevalence occurs,
it is vital to raise an alarm, for two reasons: because a change in community prevalence shifts the posterior
probabilities underlying the classifier and requires we update our classifier
predictions, and also as a matter of public health. As noted by \cite{hsu2017trend}, ``strategies are needed to
respond quickly to unexpected incidents.''

Consider the case of a dengue outbreak: the
proportion of patients with dengue will increase, but we expect the
symptoms used by \cite{tuan2015sensitivity} for classification to remain the same.  Denote by $X
\in \mathbb{R}^d$ the set of covariates for classification (such as body
temperature, vomiting, white blood cell count, etc.), and by $Y \in \{0,1\}$
a patient's true disease status.  An outbreak implies that $P(Y = 1)$ changes, but the
distributions of $X|Y=0$ and $X|Y=1$ do not.  This example---introduced by
\cite{lipton2018detecting}---constitutes a \emph{label
shift} in the distribution. 

%\vspace{-0.75em}
\subsection{Contributions}  

In this paper, we focus on nonparametric methods for detecting label shift in
sequential classification data.  The major contributions are as follows:
\begin{enumerate}
  \item We construct a
nonparametric procedure for detecting label shift changepoints in the data
distribution, and prove that it is asymptotically optimal under assumptions on the
performance of the underlying classifier.  We also demonstrate in simulations that the
procedure performs well under mild violations of both
the classifier assumptions and the label shift assumpton.

  \item We provide new, more general theoretical results about the
    performance of any nonparametric changepoint detection procedures that are
    based on likelihood ratio estimates.  Our results guarantee asymptotic
    optimality for changepoint detection when the likelihood ratio estimate
    converges in total variation distance to the true likelihood ratio.  These
    results are applicable beyond the label shift setting considered in our paper.  

  \item We demonstrate significantly improved performance of the proposed
    procedures over the current state-of-the-art in both simulation and real
    data.  For the latter, we apply our
    procedure to detect changes in dengue prevalence using real data from
    \cite{tuan2015sensitivity}.

\end{enumerate}

\noindent In Section \ref{sec:problem-method-background}, we formally
develop the problem and provide relevant background on sequential detection and
label shift. 
In Section \ref{sec:proposed-method}, we describe a simple nonparametric detection
statistic based on underlying classifier scores. Intuitively, performance of the
proposed procedure will depend on performance of the classifier.  In Section
\ref{sec:theory}, we make this relationship clear by developing new theoretical
results on the performance of a broader
class of nonparametric detection procedures. We demonstrate the efficiency of
the proposed procedure in both simulation and in an application to real dengue
data in Sections
\ref{sec:simulations} and \ref{sec:application}, in comparison to other nonparametric detection procedures.

\section{Problem and method}
\label{sec:problem-method-background}

\begin{figure}[ht]
\includegraphics[scale=0.22]{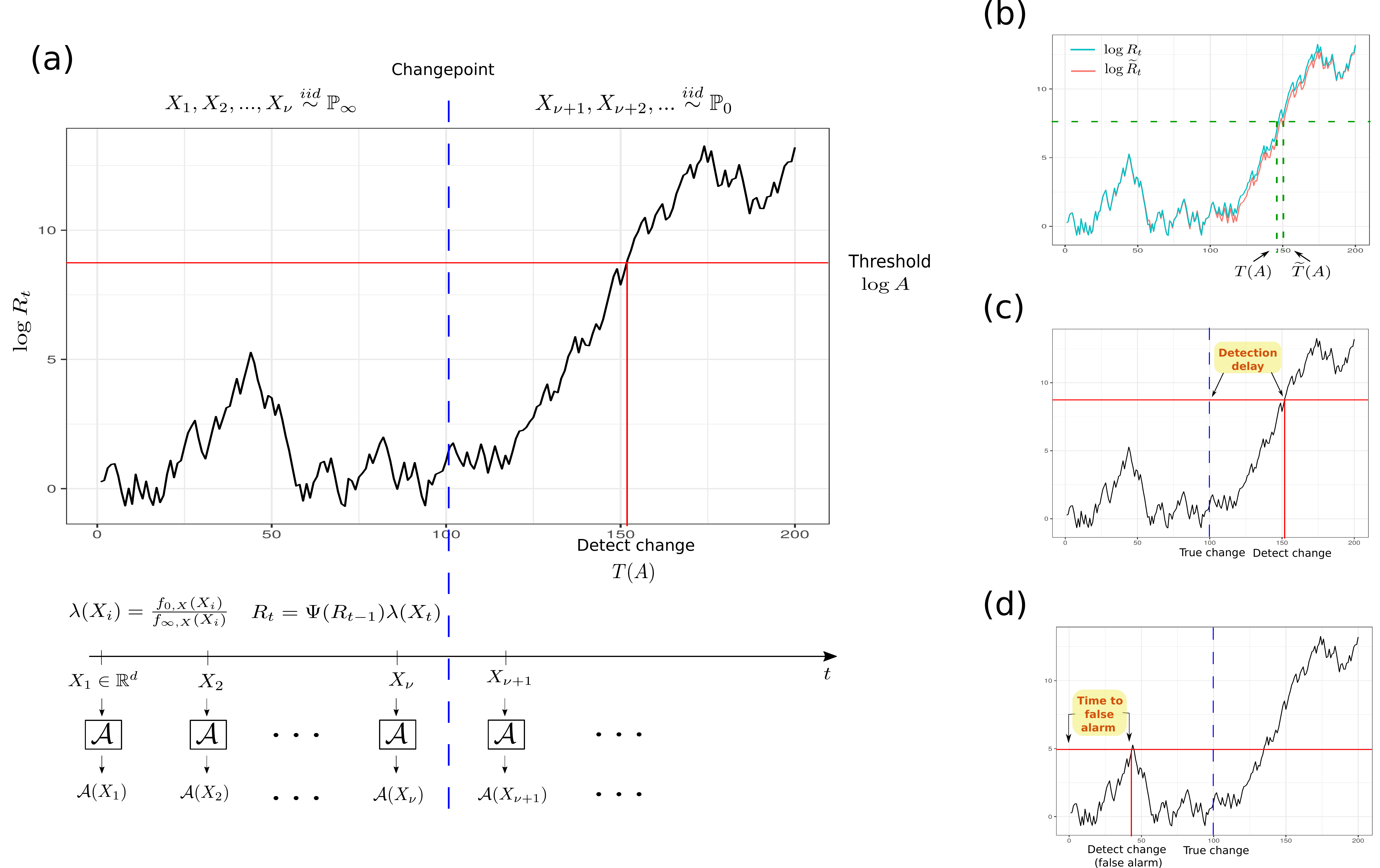}
\caption{Overview of sequential changepoint detection in the classifier setting. \textbf{(a)} Data $X_1, X_2,...$ is observed from the pre-change distribution $\mathbb{P}_\infty$ and the post-change distribution $\mathbb{P}_0$. At each time $t$, a prediction $\mathcal{A}(X_t)$ is made. If $f_{0,X}$ and $f_{\infty, X}$ are known, then a detection statistic $R_t$ can be calculated using the likelihood ratio $\lambda(X_t)$. A change is detected when $R_t \geq A$ (or equivalently $\log R_t \geq \log A$). \textbf{(b)} When the true likelihood ratio $\lambda$ is unknown, we can use an estimate $\widehat{\lambda}$ instead; $\widetilde{R}_t = \Psi(\widetilde{R}_{t-1}) \widehat{\lambda}(X_t)$ is the resulting detection statistic. When $\widehat{\lambda}$ is close to $\lambda$, the stopping times $\widetilde{T}(A)$ and $T(A)$ are also expected to be close. \textbf{(c)} When a change is detected after the true changepoint $\nu$, then $T(A) - \nu$ is the detection delay. \textbf{(d)} When $T(A) < \nu$, then we have a false alarm, and $T(A)$ is the time to false alarm.}
\label{fig:cp-overview}
\end{figure}

\subsection{Problem statement and notation}

We consider a sequential classification setting with unobserved labels, where
feature vectors $X_1, X_2, X_3,... \in \mathbb{R}^d$ arrive sequentially, but
the associated labels $Y_i \in \{0,1\}$ are \emph{unobserved}. In our dengue
example, $X_i$ represents diagnostic measurements like temperature and white
blood cell count, while $Y_i$ represents true dengue status.  We assume that a
classifier, $\mathcal{A}(\cdot)$, has been trained on a separate set of
training observations $(X_1', Y_1'),...,(X_m', Y_m')$ and is used to predict
the unobserved labels $Y_i$.  

At some time $\nu \geq 0$ in this sequence, called the \textit{changepoint}, the distribution of $(X_i, Y_i)$
changes.  We notate the pre-change distribution as $\mathbb{P}_{\infty}$ and
the post-change distribution as $\mathbb{P}_0$, such that $(X_1', Y_1')$,...,$(X_m', Y_m')$, $(X_1, Y_1)$,...,$(X_\nu, Y_\nu) \overset{iid}{\sim} \mathbb{P}_\infty$ and $(X_{\nu + 1}, Y_{\nu + 1}), (X_{\nu + 2}, Y_{\nu + 2}),... 
\overset{iid}{\sim} \mathbb{P}_0$.
Our aim is to detect the change in the distribution of $(X_i,
Y_i)$ as quickly as possible, using the observed sequence $X_i$.

\begin{remark}
Throughout the paper we use the subscripts $\infty$ and 0 for pre- and post-change
quantities respectively, to be consistent with the sequential changepoint
detection literature. The motivation is that $\nu = \infty$ indicates the
change never occurs, so data is from the pre-change distribution, while $\nu =
0$ indicates the change occurs before we observe any data, so data is from the
post-change distribution.
When context is clear we will let
$\mathbb{P}_\infty$ and $\mathbb{P}_0$ denote general pre- and post-change
distributions, so for example $(X_i, Y_i) \sim \mathbb{P}_\infty$ and $X_i \sim
\mathbb{P}_\infty$ both indicate data drawn before a change occurs.
\end{remark}

The general problem of classification under a changed distribution
has been studied
extensively in the literature. Because arbitrary changes to high-dimensional classification data may be impossible to
correct or detect, it is standard to make additional assumptions on the nature of the
change. Because it frequently arises in practice, we will focus on the \textbf{label shift} setting
\citep{saerens2002adjusting, storkey2009training}, which has received recent
attention in the machine learning literature
\citep{ackerman2020sequential, azizzadenesheli2019regularized, lipton2018detecting,
rabanser2019failing}. Label shift assumes that the marginal distribution of
$Y_i$ changes, but the conditional distribution of $X_i|Y_i$ does not:

\begin{definition}[Label shift]
Let $f_{\infty, X, Y}$, $f_{\infty, Y}$, and $f_{\infty, X|Y=y}$ denote the densities/mass functions of $(X, Y)$, $Y$, and $X|Y=y$ respectively, under $\mathbb{P}_\infty$. Similarly define $f_{0, X, Y}$, $f_{0, Y}$, and $f_{0, X|Y=y}$. The label shift assumption is that $f_{0, X|Y=y} \equiv f_{\infty, X|Y=y}$ for all $y$, so
\begin{align}
\label{eq:label-shift}
f_{0, X, Y}(x,y) = f_{0, Y}(y)f_{0, X|Y=y}(x) = f_{0, Y}(y)f_{\infty, X|Y=y}(x) \hspace{1cm} \forall x,y.
\end{align}
\end{definition}

Label shift is simply a change in the mixing proportion for the class
distributions $X|Y=0$ and $X|Y=1$. Since the conditional distribution of
$X|Y=y$ remains the same, the conditional distribution of the classifier
predictions, $\mathcal{A}(X)|Y=y$, does as well. As we show in the following sections, the classifier predictions $\mathcal{A}(X_i)$ can be further leveraged to improve changepoint detection in the label shift setting.

%\vspace{-0.75em}
\subsection{Proposed method}
\label{sec:proposed-method}

Let $(X_1', Y_1'),...,(X_m', Y_m') \overset{iid}{\sim} \mathbb{P}_\infty$ denote our labeled training set, used to train the classifier $\mathcal{A}(\cdot)$. Classifiers typically predict either the probability of a positive case $\mathbb{P}(Y_i = 1 | X_i)$ or the label $Y_i$, and so we assume $\mathcal{A}(X_i) \in [0,1]$ is a predicted probability, or $\mathcal{A}(X_i) \in \{0,1\}$ is a predicted label. 

To detect a change in the unlabeled sequence $X_1, X_2, ...$, many changepoint detection procedures use a recursive detection statistic $R_t^x = \Psi(R_{t-1})\lambda(X_t)$, where $\lambda(X_t) = f_{0,X}(X_t)/f_{\infty, X}(X_t)$ is the likelihood ratio at time $t$, $\Psi$ is an update function, and the initial value is $R_0^x = x$ (Figure \ref{fig:cp-overview}(a)). For example, the classical CUSUM procedure has $\Psi(r) = \max\{1, r\}$ and $x = 1$, while the Shiryaev-Roberts procedure has $\Psi(r) = 1 + r$ and $x = 0$ \citep{polunchenko2012state}. A change is detected when $R_t^x$ crosses a pre-specified threshold $A > R_0^x$, with stopping time $T^x(A) = \inf \{t \geq 1: R_t^x \geq A\}$.

In most applications, $\lambda = f_{0,X}/f_{\infty,X}$ is unknown. However, under the label shift assumption, we can rewrite the likelihood ratio as
\begin{align}
\frac{f_{0, X}(X_i)}{f_{\infty, X}(X_i)} = \left( \frac{\pi_0}{\pi_\infty} - \frac{1 - \pi_0}{1 - \pi_\infty} \right) \mathbb{P}_\infty(Y_i = 1 | X_i) + \frac{1 - \pi_0}{1 - \pi_\infty}, \label{eq:scoreratio}
\end{align}
with pre- and post-change proportions $\pi_\infty = \mathbb{P}_\infty(Y = 1)$ and $\pi_0 = \mathbb{P}_0(Y = 1)$. Noting that $\mathbb{P}_\infty(Y = 1 | X = x)$ is a typical estimand in classification, we define the estimated likelihood ratio $\widehat{\lambda}_{\mathcal{A},m}$ by
\begin{align}
\label{eq:est-lr}
\widehat{\lambda}_{\mathcal{A}, m}(x) = \left( \frac{\pi_0}{\pi_\infty} - \frac{1 - \pi_0}{1 - \pi_\infty} \right) \mathcal{A}(x) + \frac{1 - \pi_0}{1 - \pi_\infty},
\end{align}
where the subscripts $\mathcal{A}$ and $m$ denote dependence on the classifier and the size of the training set. Our method detects label shift with the nonparametric detection statistic $\widetilde{R}_t^x$ and resulting stopping time $\widetilde{T}^x(A)$, where $\widetilde{R}_0^x = x$ and $\widetilde{R}_t^x = \Psi(\widetilde{R}_{t-1}^x)\widehat{\lambda}_{\mathcal{A}, m}(X_t)$, and $\widetilde{T}^x(A) = \inf \{t \geq 1 : \widetilde{R}_t^x \geq A\}$ (Figure \ref{fig:cp-overview}(b)). When the pre- and post-change proportions $\pi_\infty$ and $\pi_0$ are known, performance of our label shift detection method depends on performance of the classifier $\mathcal{A}$, which we formalize in Section \ref{sec:theory} and illustrate with simulations in Section \ref{sec:simulations}. When $\pi_0$ is unknown, we use a mixing procedure that integrates \eqref{eq:est-lr} over possible values of $\pi_0$ (as we have access to labeled pre-change training data, the assumption that $\pi_\infty$ is known is reasonable). We discuss this mixing procedure in the following section.

\begin{remark}
Using classifier predictions to estimate the likelihood ratio is natural in the label shift setting, as a classifier is already constructed and being applied to make predictions for new data. However, an advantage of the label shift setting is that it supports a variety of other approaches to likelihood ratio estimation. For example, approaches like kernel mean matching \citep{gretton2009covariate} and uLSIF \citep{kanamori2009least} rely on both pre- and post-change data; under the label shift assumption, a post-change sample can be generated by re-sampling or re-weighting the training data $(X_1', Y_1'),...,(X_m', Y_m')$ when $\pi_0$ is known. We compare this approach to our classifier-based likelihood ratio estimate in Section \ref{sec:simulations}.
\end{remark}

%\vspace{-0.75em}
\subsubsection{Changepoint detection with unknown $\pi_0$}
\label{sec:unknown-pi0}

As part of training our classifier $\mathcal{A}$, we have access to labeled pre-change training data \\ $(X_1', Y_1'),...,(X_m', Y_m')$, but it is less common to have a sample of post-change data, and so the post-change parameter $\pi_0$ is often unknown. To overcome an unknown $\pi_0$, we mix over a set $\Pi_0 \subset [0, 1]$ of potential values for the
post-change parameter, with a weight distribution $w$. Here we are inspired by
the work of \citet{lai1998information}, which deals with the computational complexity involved in the integration by considering a window-limited approach that uses only a fixed number of the most
recent observations. Let $\Pi_0$ be the set of possible values for $\pi_0$, and
let $w(\pi_0)$ be a density on $\Pi_0$. Each potential $\pi_0$ results in a
different likelihood ratio function $\lambda_{\pi_0}$. Lai defines a
CUSUM-type mixture stopping rule with detection statistic $R_{t, w}$ and
stopping time $T_w(A)$ \citep{lai1998information}:
\begin{align}
\label{eq:lai-mix}
R_{t, w} = \max \limits_{t - m_\alpha \leq k \leq t} \int \limits_{\Pi_0} \prod \limits_{i=k}^t \lambda_{\pi_0}(X_i) w(\pi_0) d\pi_0 \hspace{1cm}
T_w(A) = \inf \{t \geq 1: R_{t,w} \geq A \},
\end{align} 
where $m_\alpha$ is the window size. In our label shift setting, we have
\begin{align}
\lambda_{\pi_0}(x) = \frac{\pi_0 f_{\infty, X|Y=1}(x) + (1 - \pi_0) f_{\infty, X|Y=0}(x)}{\pi_\infty f_{\infty, X|Y=1}(x) + (1 - \pi_\infty) f_{\infty, X|Y=0}(x)}.
\end{align}For each $\pi_0$, we replace $\lambda_{\pi_0}$ with its estimate $\widehat{\lambda}_{\pi_0, \mathcal{A}, m}$ from \eqref{eq:est-lr}, yielding the detection statistic $\widetilde{R}_{t, w}$ and stopping time $\widetilde{T}_w(A)$:
\begin{align}
  \widetilde{R}_{t, w} = \max \limits_{t - m_\alpha \leq k \leq t} \int \limits_{\Pi_0} \prod \limits_{i=k}^t \widehat{\lambda}_{\pi_0, \mathcal{A}, m}(X_i) w(\pi_0) d\pi_0 \hspace{1cm}
\widetilde{T}_w(A) = \inf \{t \geq 1: \widetilde{R}_{t,w} \geq A \}.\label{eq:rtw}
\end{align}
When $\mathcal{A}(X_i)$ is a good estimate of $\mathbb{P}_\infty(Y_i = 1 | X_i)$, we expect $\widetilde{T}_w(A)$ \eqref{eq:rtw} to be close to $T_w(A)$ \eqref{eq:lai-mix}. We formalize this statement in Section \ref{sec:theory-unknown-pi0}, and we show that our estimated mixing procedure performs well on real dengue data in Section \ref{sec:application}.

An alternative to mixing over $\Pi_0$ is to maximize over possible values of $\pi_0$ at each time step. This is the generalized likelihood ratio (GLR) approach, and has also been studied in previous research (see, e.g., \citet{siegmund1995using}). For exponential families, some optimality properties of the GLR have been shown, but it is typically harder to control the average run length to false alarm \citep{tartakovsky2014sequential}. Another option is to perform detection with a worst-case $\pi_0^* \in \Pi_0$ \citep{unnikrishnan2011minimax}, which provides a worst-case bound on detection delay. 

%\vspace{-0.75em}
\subsection{Background and related literature}
\label{sec:background}

\subsubsection{Label shift testing} 

Non-sequential two-stample tests between training and test data have been
proposed for detecting label shift between batches of data.
\citet{saerens2002adjusting} propose a likelihood ratio 
test, based on expectation-maximization.
In \citet{lipton2018detecting}, the
authors note that label shift implies a change in the distribution of
classifier predictions, and therefore use a two-sample test directly on the
training and test set predictions to detecting the change. This is expanded by
\citet{rabanser2019failing}, who recommend tests for label shift as a general method for detecting distributional changes, even if the label shift assumption is not met. For sequential detection of label shift, \citet{ackerman2020sequential} implement the nonparametric detection procedure in \citet{ross2012two} based on repeated Cramer--von-Mises tests for a change in distribution, using the \texttt{cpm} package \citep{cpm} in \texttt{R}. Like \citet{ackerman2020sequential}, we consider the problem of detecting label shift in a sequential setting, rather than a batch setting. In the sequential setting, label shift detection is an example of the classic problem of sequential changepoint detection, and we apply nonparametric sequential detection tools to the label shift problem. In contrast to \citet{ackerman2020sequential}, we use classifier predictions and the label shift assumption to directly approximate the likelihood ratio for sequential detection, which allows us to construct asymptotically optimal nonparametric label shift detection procedures.

%\vspace{-0.75em}
\subsubsection{Nonparametric detection procedures} 

Because the pre- and post-change
data distributions are rarely known in practice, a variety of nonparametric
detection procedures have been proposed. For example, several
authors have adapted nonparametric hypothesis tests to the changepoint
detection problem, such as Kolmogorov-Smirnov tests \citep{madrid2019sequential},
Cramer-von-Mises tests \citep{ross2012two}, and graph-based nearest-neighbors
tests \citep{chen2019sequential, chu2018sequential}. Another common approach is
to replace the likelihood ratio $\lambda$ with an estimate $\widehat{\lambda}$.
Often, this estimate $\widehat{\lambda}$ is 
constructed to detect specific types of expected changes, such as shifts in the mean
or variance \citep{brodsky1993nonparametric, brodsky2000diagnosis,
tartakovsky2012efficient, tartakovsky2006detection,  tartakovsky2006novel}, or
a change to a stochastically larger/smaller distribution
\citep{bell1994efficient, gordon1994efficient, gordon1995robust,
mcdonald1990cusum}. 

Other authors employ nonparametric density ratio estimates that utilize samples from both the pre- and post-change distributions $\mathbb{P}_\infty$ and $\mathbb{P}_0$. For example, \citet{baron2000nonparametric} estimates the post-change
distribution sequentially with a histogram density estimator, while another approach is to choose the ratio that maximizes an estimate of the divergence between the pre- and post-change distributions \citep{nguyen2010estimating, kanamori2009least, kawahara2009change, sugiyama2008direct, liu2013change}. Similarly, the kernel mean matching approach \citep{gretton2009covariate, yu2012analysis} estimates the ratio by matching moments after mapping into a reproducing kernel Hilbert space (RKHS), while \cite{bickel2009discriminative} propose training a classifier to predict whether data comes from the pre- or post-change distribution.

In the label shift case, the difficulty of having samples from the post-change
distribution is reduced to knowing the post-change parameter $\pi_0$ (see
\eqref{eq:label-shift}). In the following section, we propose a simple estimate
of the likelihood ratio as a linear function of the classifier scores
$\mathcal{A}(X_i)$. An advantage of this method is that the existing classifier
can be used without additional estimation, and performance of the detection procedure is directly related to performance of the classifier. In addition, when $\pi_0$ is unknown it is easy to calculate the likelihood ratio over a range of potential values $\Pi_0$, without having to re-estimate the ratio each time (see Section \ref{sec:unknown-pi0}).

\subsubsection{Operating characteristics}  

The performance of sequential detection procedures, with stopping time $T^x(A)$ at threshold $A$, is typically assessed by two operating characteristics, the \textit{average time to false alarm} $\mathbb{E}_\infty(T^x(A))$ (also called the \textit{average run length}, or ARL), and the \textit{average detection delay} $\mathbb{E}_0(T^x(A))$, which are expected stopping times under the pre- and post-change distributions respectively (Figure \ref{fig:cp-overview}(c) and \ref{fig:cp-overview}(d)). The goal is to minimize the average detection delay, subject to a lower bound on the average time to false alarm, and the CUSUM and Shiryaev-Roberts procedures are known to be optimal or approximately optimal for this problem \citep{lorden1971procedures, moustakides1986optimal, tartakovsky2012third}. We therefore compare average detection delay and average time to false alarm as a way to assess procedures in this manuscript.

\section{Operating characteristics of nonparametric detection procedures}
\label{sec:theory}

In Section \ref{sec:proposed-method}, we introduced a nonparametric method for detecting label shift in sequential data. Our method relies on the classifier predictions $\mathcal{A}(X_i)$, which can be used to directly estimate the likelihood ratio under the label shift assumption (see \eqref{eq:est-lr} and \eqref{eq:rtw}). Our nonparametric procedure defines a stopping time $\widetilde{T}^x(A)$ (when the post-change proportion $\pi_0$ is known) or $\widetilde{T}_w(A)$ (when $\pi_0$ is unknown). Our stopping time is an approximation of the \textit{optimal} stopping time $T^x(A)$ (or $T_w(A)$) which depends on the true likelihood ratio $\lambda = f_{0,X}/f_{\infty, X}$. 

Naturally, we expect that the better $\mathcal{A}(x)$ estimates $\mathbb{P}_\infty(Y = 1 | X = x)$, the closer $\widetilde{T}^x(A)$ and $\widetilde{T}_w(A)$ are to $T^x(A)$ and $T_w(A)$. The purpose of this section is to formalize how detection performance depends on the classifier $\mathcal{A}$, which provides insight on the performance of our method and on selecting a classifier for changepoint detection. In Section \ref{sec:define-performance}, we define our measures of detection performance, based on convergence of expected stopping times. In Section \ref{sec:continuous-convergence}, we provide conditions under which this convergence holds, and in Theorem \ref{thm:inf-norm-rate-conv} we give an upper bound on the rate of convergence which depends directly on the likelihood ratio estimate $\widehat{\lambda}$. Theorem \ref{thm:inf-norm-rate-conv} applies beyond the label shift setting to any estimate of the likelihood ratio, so we state it in generality and discuss our results in the context of label shift. For the special case of label shift, Theorem \ref{thm:inf-norm-rate-conv} requires $\pi_0$ to be known, and so we generalize our results to an unknown $\pi_0$ in Section \ref{sec:theory-unknown-pi0}. Finally, Theorem \ref{thm:inf-norm-rate-conv} assumes that $\mathcal{A}(X)$ is continuous, which occurs when $\mathcal{A}(x) \in [0, 1]$ estimates $\mathbb{P}_\infty(Y = 1 | X = x)$. In Section \ref{sec:separable-dists}, we describe conditions under which similar convergence results hold for binary classifiers $\mathcal{A}(X) \in \{0, 1\}$. Examples illustrating our results are provided in Section \ref{sec:examples}.

\subsection{Assessing performance of a nonparametric detection procedure}
\label{sec:define-performance}

As discussed in Section \ref{sec:background}, it is common to assess detection performance with the expected stopping times $\mathbb{E}_\infty(T)$ and $\mathbb{E}_0(T)$ under the pre- and post-change distributions. Previous results on the performance of nonparametric detection procedures have typically focused on the relative efficiency of estimated procedures compared to optimal performance (see, e.g., \cite{bell1994efficient, unnikrishnan2011minimax}), by examining the limiting behavior of the ratio $\mathbb{E}_0(\widetilde{T}^x(\widetilde{A}))/\mathbb{E}_0(T^x(A))$ as $\widetilde{A}, A \to \infty$, where $\widetilde{A}$ is chosen so that $\mathbb{E}_\infty(\widetilde{T}^x(\widetilde{A})) = \mathbb{E}_\infty(T^x(A))$.

While relative efficiency is useful for comparing detection procedures, and it is natural to compare detection delay at the same average time to false alarm, formal results are typically asymptotic in the thresholds $A$ and $\widetilde{A}$. As an alternative, we consider a single fixed threshold $A$ and the differences $|\mathbb{E}_i(\widetilde{T}^x(A)) - \mathbb{E}_i(T^x(A))|$, $i \in \{0, \infty\}$. To emphasize the dependence of $\mathbb{E}_i(\widetilde{T}^x(A))$ on $\widehat{\lambda}$, we will write $\mathbb{E}_i(\widetilde{T}^x(A) | \widehat{\lambda}_m)$, where the subscript $m$ is used to show the dependence of $\widehat{\lambda}_m$ on the size of the training set. Ideally, $\mathbb{E}_i(\widetilde{T}^x(A) | \widehat{\lambda}_m) \overset{p}{\to} \mathbb{E}_i(T^x(A))$ as $m \to \infty$. In the following sections, we provide conditions under which this convergence in probability holds, and we provide upper bounds on the rate of convergence.

\subsection{Convergence for continuous likelihood ratios}
\label{sec:continuous-convergence}

We first consider the case when $\lambda(X)$ and $\widehat{\lambda}_m(X)$ are continuous, which is a common assumption in the changepoint detection literature. Under assumption (A1) - (A5) below, the convergence of $|\mathbb{E}_i(\widetilde{T}^x(A)|\widehat{\lambda}_m) - \mathbb{E}_i(T^x(A))|$ depends on the total variation distance between the distributions of $\lambda(X)$ and $\widehat{\lambda}_m(X)$. Our main result is Theorem \ref{thm:inf-norm-rate-conv}, which is not restricted to the label shift setting; let $\widehat{\lambda}_m$ be any likelihood ratio estimate from a training sample of size $m$ (not just \eqref{eq:est-lr}).

\begin{enumerate}
\item[(A1)] The detection procedure is defined by a detection statistic $R_t^x$ and a stopping time $T^x(A) = \inf \{t \geq 1: R_t^x \geq A\}$, with $R_t^x = \Psi(R_{t-1}^x) \lambda(X_t)$ and $R_0^x = x$. Likewise, the estimated detection procedure is defined by $\widetilde{T}^x(A) = \inf \{t \geq 1: \widetilde{R}_t^x(A) \geq A\}$, with $\widetilde{R}_t^x = \Psi(\widetilde{R}_{t-1}^x) \widehat{\lambda}_m(X_t)$ and $\widetilde{R}_0^x = x$.

\item[(A2)] The distributions of $\lambda(X)$ and $\widehat{\lambda}_m(X)$ are continuous, with $X \sim \mathbb{P}_i$.

\item[(A3)] The densities $f^i_\lambda$ and $f^i_{\widehat{\lambda}_m}$ of $\lambda(X)$ and $\widehat{\lambda}_m(X)$ under $\mathbb{P}_i$ satisfy the following assumptions:
\begin{enumerate}
\item $\int \limits_0^A (f^i_{\lambda}(s))^2 ds  \ < \infty$ and $\int \limits_0^A (f^i_{\widehat{\lambda}_m}(s))^2 ds  \ < \infty$,

\item $\lim \limits_{\varepsilon \to 0} \sup \limits_{x, z \in [0, A]; |x - z| < \varepsilon} \int \limits_0^A |f^i_\lambda \left( \frac{y}{\Psi(x)} \right) - f^i_\lambda \left( \frac{y}{\Psi(z)} \right)| dy \ = 0$ (and likewise for $f^i_{\widehat{\lambda}_m}$).
\end{enumerate}

\item[(A4)] The functions $\Psi(r)$ and $1/\Psi(r)$ are Lipschitz continuous on $[0, A]$, and $\Psi(r) \geq 1$ for all $r$.

\item[(A5)] The total variation distance $TV(f^i_\lambda, f^i_{\widehat{\lambda}_m}) = \int |f^i_\lambda(s) - f^i_{\widehat{\lambda}_m}(s)| ds \ \overset{p}{\to} 0$ as $m \to \infty$.
\end{enumerate}

Assumption (A1) holds for standard detection procedures like CUSUM, Shiryaev-Roberts, and variants. Note that the assumption that $\lambda(X)$ is continuous is common in the changepoint detection literature, and is relied on by many approaches to calculating or approximating the expected stopping time of a detection procedure. In general we are interested in $|\mathbb{E}_i(\widetilde{T}^x(A)|\widehat{\lambda}_m) - \mathbb{E}_i(T^x(A))|$ for both $i = 0$ and $i = \infty$, which means (A2) requires $X$ is continuous and has the same support under $\mathbb{P}_\infty$ and $\mathbb{P}_0$, such as the ratio for two normal distributions (see Example \ref{ex:normal-shift} below). Assumption (A4) holds for common $\Psi$, such as $\Psi(r) = \max\{1, r\}$ (CUSUM) and $\Psi(r) = 1 + r$ (Shiryaev-Roberts). Assumption (A3) can be thought of as a requirement that the distributions of $\lambda(X)$ and $\widehat{\lambda}_m(X)$ are not close to having any point masses. A sufficient condition for (A3) is that the densities $f^i_\lambda$ and $f^i_{\widehat{\lambda}_m}$ are Lipschitz continuous.

\begin{theorem}
\label{thm:inf-norm-rate-conv}
Suppose that assumptions (A1) - (A5) hold. Then 
\begin{align}
|\mathbb{E}_i(\widetilde{T}^x(A) | \widehat{\lambda}_m) - \mathbb{E}_i(T^x(A))| \leq O_P(TV(f^i_\lambda, f^i_{\widehat{\lambda}_m})).
\end{align}
If, in addition, $f_\lambda^i$ is bounded, then
\begin{align}
\label{eq:bounded-density-rate}
|\mathbb{E}_i(\widetilde{T}^x(A) | \widehat{\lambda}_m) - \mathbb{E}_i(T^x(A))| \leq O_P(\mathbb{E}_i( |\widehat{\lambda}_m(X) - \lambda(X)| ) + O_P(||F^i_{\widehat{\lambda}_m} - F^i_\lambda||_\infty ).
\end{align}
\end{theorem}

\begin{proof}
See the Supplementary Materials, Section \ref{sec:theorem1-proof}.
\end{proof}

\subsubsection{Application to label shift}

The estimate $\widehat{\lambda}_{\mathcal{A},m}$ from \eqref{eq:est-lr} can be reasonably assumed continuous when $\mathcal{A}(X) \in [0,1]$ is a predicted probability. The bound in Theorem \ref{thm:inf-norm-rate-conv} depends on the distribution of the likelihood ratio, which is intuitive as our detection procedure includes a random draw from this distribution at each step (see (A1)). As we can see from \eqref{eq:bounded-density-rate}, when $f_\lambda^i$ is bounded, convergence depends directly on the $L_1$ rate of convergence of the likelihood ratio estimate $\widehat{\lambda}_m$. In the classifier setting, this is equivalent to the $L_1$ convergence of the classifier scores $\mathcal{A}(X)$ to the true probabilities $\mathbb{P}_\infty(Y = 1 | X)$, illustrating that performance of the changepoint detection procedure depends on performance of the classifier. If parametric assumptions are made, then we can expect parametric rates of convergence for the operating characteristics, as discussed in Example \ref{ex:normal-shift} and Example \ref{ex:lda-label-shift} below. Note that without parametric assumptions, the optimal $L_1$ rate of convergence is $m^{-p/(2p + d)}$ if $\lambda$ is a $p$-times differentiable function \citep{stone1982optimal}, though not all classifiers will converge. In general, results for nonparametric regression are more common for $L_2$ convergence rather than $L_1$ convergence, and the same is true for other methods of density ratio estimation -- for example, \cite{nguyen2010estimating} provide conditions under which a density ratio estimate from divergence maximization converges in Hellinger distance. Finally, note that convergence in \eqref{eq:bounded-density-rate} requires $\pi_0$ to be known for the estimate $\widehat{\lambda}_{\mathcal{A}, m}$ from \eqref{eq:est-lr}. In the next section, we focus specifically on the label shift setting and consider the situation where $\pi_0$ is unknown, to generalize convergence results to our mixture procedure in \eqref{eq:rtw}.

\begin{remark}
Consider the classifier label shift setting, with unlabeled data, where the true likelihood ratio $\lambda(x) = f_{0,X}(x)/f_{\infty, X}(x)$ is given by \eqref{eq:scoreratio}. Suppose that $\lambda(X)$ is continuous, but we use a binary classifier with $\mathcal{A}(X) \in \{0, 1\}$ (for example, by thresholding predicted probabilities), and estimate $\lambda$ with $\widehat{\lambda}_{\mathcal{A},m}$ as in \eqref{eq:est-lr}. Because the binary predictions $\mathcal{A}(X)$ will never converge to the true probabilities $\mathbb{P}_\infty(Y=1|X)$, then $\mathbb{E}_i(\widetilde{T}^x(A)|\widehat{\lambda}_{\mathcal{A},m})$ won't converge to $\mathbb{E}_i(T^x(A))$. Therefore, if $\mathbb{P}_\infty(Y=1|X)$ is expected to be a continuous function of $X$, it is better to use predicted probabilities, rather than binary predictions, for changepoint detection. Binary predictions are only suitable when the classes $X|Y = 0$ and $X|Y = 1$ are separable, as we discuss below in Section \ref{sec:separable-dists}.
\end{remark}

\subsection{Convergence of label shift detection with unknown $\pi_0$}
\label{sec:theory-unknown-pi0}

In the label shift setting, Theorem \ref{thm:inf-norm-rate-conv} applies when $\lambda(X)$ and $\widehat{\lambda}_{\mathcal{A}, m}(X)$ are continuous, and $\pi_0$ is known or can be consistently estimated. However, as discussed in Section \ref{sec:unknown-pi0}, we typically expect $\pi_0$ to be unknown in practice. Here we explicitly consider the classifier label shift setting, with likelihood ratio estimate $\widehat{\lambda}_{\pi_0, \mathcal{A}, m}$ for each potential $\pi_0$ given by \eqref{eq:est-lr}. Note that an advantage of this likelihood ratio estimate is that mixing over $\Pi_0$ is simple: we need only change $\pi_0$ in \eqref{eq:est-lr}. In contrast, two-sample ratio estimation procedures like those discussed in \cite{gretton2009covariate}, \cite{nguyen2010estimating}, and \cite{kanamori2009least} would require a new post-change sample and a new likelihood ratio estimate to be calculated for each $\pi_0$.

Since $\pi_0$ is unknown, we apply the CUSUM-type mixture stopping rule \eqref{eq:rtw} discussed in Section \ref{sec:unknown-pi0}. As $\widetilde{R}_{t,w}$ cannot be written recursively, the proof techniques of Theorem \ref{thm:inf-norm-rate-conv} do not apply, but it is still possible to show that $\mathbb{E}_i(\widetilde{T}_w(A) | \mathcal{A})$ is consistent for $\mathbb{E}_i(T_w(A))$ as $m \to \infty$, though we lose the ability to provide a rate of convergence:

\begin{theorem}
\label{thm:dens-est-mix}
Let $\Pi_0 = [a, b]$ where $0 < a \leq b < 1$, and suppose there exist sets $\mathcal{S}_c$ indexed by $c > 0$, such that $\mathcal{S}_{c_1} \subseteq \mathcal{S}_{c_2}$ when $c_1 > c_2$, $\lim \limits_{c \to 0} \mathbb{P}_i(X \in \mathcal{S}_c) = 1$, and such that $\sup \limits_{\substack{x \in \mathcal{S}_c \\ \pi_0 \in \Pi_0}} | \widehat{\lambda}_{\pi_0, \mathcal{A}, m}(x) - \lambda_{\pi_0}(x)| \overset{p}{\to} 0$ for all $c$. Furthermore, assume that $\mathbb{E}_i(T_w(A))$ is a continuous function of $A$. Then for all $A > 0$, $\mathbb{E}_i(\widetilde{T}_w(A) | \mathcal{A}) \overset{p}{\to} \mathbb{E}_i(T_w(A))$.
\end{theorem}
\begin{proof}
See the Supplementary Materials, Section \ref{sec:thm-2-proof}.
\end{proof}

In practice, the mixture procedure in \eqref{eq:rtw} often performs quite close to the procedure with estimated likelihood ratio \eqref{eq:est-lr}, and performance improves when $\Pi_0$ can be narrowed. We will see in Section \ref{sec:application} that the mixture procedure outperforms other nonparametric detection procedures which don't require knowledge of $\pi_0$ for detecting a change in dengue prevalence.

\subsection{Label shift detection with binary predictions: separable class distributions}
\label{sec:separable-dists}

As discussed above, if $\lambda(X)$ is continuous then it makes sense to use predicted probabilities $\mathcal{A}(X) \in [0, 1]$ to construct the likelihood ratio estimate $\widehat{\lambda}_{\mathcal{A}, m}(X)$. However, it is common to use binary classifiers with $\mathcal{A}(X) \in \{0, 1\}$. In this section, we provide a convergence result for label shift detection based on binary predictions, which depends on separable class distributions.

If $X|Y=0$ and $X|Y=1$ are separable, then optimal changepoint detection with the unlabeled data $X_1, X_2, X_3,...$ is equivalent to optimal changepoint detection with the labeled data $(X_1, Y_1), (X_2, Y_2), (X_3, Y_3),...$, with likelihood ratio 
\begin{align}
\label{eq:true-binary-lr}
\frac{f_{0, X,Y}(X_i, Y_i)}{f_{\infty, X, Y}(X_i, Y_i)} = \left( \frac{\pi_0}{\pi_\infty} - \frac{1 - \pi_0}{1 - \pi_\infty} \right) Y_i + \frac{1 - \pi_0}{1 - \pi_\infty}.
\end{align}
Assuming the classifier $\mathcal{A}$ can separate the two classes given enough training data, then the nonparametric detection procedure is asymptotically optimal, and convergence of $|\mathbb{E}_i(\widetilde{T}^x(A) | \widehat{\lambda}_m) - \mathbb{E}_i(T^x(A))|$ depends on the sensitivity and specificity of $\mathcal{A}$:

\begin{corollary}
\label{cor:binary-convergence}
Suppose that $\pi_\infty$ and $\pi_0$ are known, and furthermore that $\log(\pi_0/\pi_\infty)$ and $\log((1 - \pi_0)/(1 - \pi_\infty))$ are rational. The optimal detection procedure observes $(X_1, Y_1), (X_2, Y_2),...$, with likelihood ratio $\lambda(X_i, Y_i)$ given by \eqref{eq:true-binary-lr}, and detection statistic $R_t^x = \max\{1, R_{t-1}^x\} \lambda(X_t, Y_t)$. When the labels $Y_i$ are unobserved, the nonparametric detection procedure uses a binary classifier $\mathcal{A}$, with $\mathcal{A}(X) \in \{0,1\}$; the likelihood ratio $\widehat{\lambda}_{\mathcal{A},m}$ is given by \eqref{eq:est-lr}, and the detection statistic is $\widetilde{R}_t^x = \max\{1, \widetilde{R}_{t-1}^x\} \widehat{\lambda}_{\mathcal{A},m}(X_t)$. Then, if $\mathbb{P}_i(\mathcal{A}(X) = 1 | Y = 1, \mathcal{A}) \overset{p}{\to} 1$ and $\mathbb{P}_i(\mathcal{A}(X) = 0 | Y = 0, \mathcal{A}) \overset{p}{\to} 1$ as $m \to \infty$,
\begin{align}
\label{eq:binary-convergence-rate}
|\mathbb{E}_i(\widetilde{T}^x(A) | \widehat{\lambda}_m) - \mathbb{E}_i(T^x(A))| \leq O_P(\mathbb{P}_i(\mathcal{A}(X) = 0 | Y = 1, \mathcal{A}) + \mathbb{P}_i(\mathcal{A}(X) = 1 | Y = 0, \mathcal{A})).
\end{align}
\end{corollary}
\begin{proof}
See the Supplementary Materials, Section \ref{sec:cor-1-proof}.
\end{proof}

The detection procedure in Corollary \ref{cor:binary-convergence} is simply a Bernoulli CUSUM procedure, and the right hand side of \eqref{eq:binary-convergence-rate} depends on the convergence of the specificity and sensitivity of the classifier $\mathcal{A}$, which requires that the two distributions $X|Y=0$ and $X|Y=1$ can be separated. This is a stronger requirement than in Theorem \ref{thm:inf-norm-rate-conv}: consistently estimating probabilities $\mathbb{P}_i(Y = 1 | X)$ can be possible even when consistently estimating labels is impossible. In practice, $X|Y=0$ and $X|Y=1$ are rarely perfectly separable, but Corollary \ref{cor:binary-convergence} is still helpful for seeing the relationship between classifier performance and changepoint performance.

The assumption that $\log(\pi_0/\pi_\infty)$ and $\log((1 - \pi_0)/(1 - \pi_\infty))$ are rational is needed to ensure that $R_t^x$ and $\widetilde{R}_t^x$ are Markov processes on the same state space, as is the restriction of $\Psi$ to the CUSUM $\Psi(x) = \max\{1,x\}$. In practice, \cite{reynolds1999cusum} show that the expected stopping time is close when $\log \lambda(X)$ is not rational but a rational approximation is used.

\subsection{Examples}
\label{sec:examples}

To help illustrate our theoretical results, we consider several specific change detection scenarios, which provide concrete examples of our results in action. First, we consider a simple univariate parametric setting. When we have a parametric model, we hope that the rate of convergence for our detection procedure is the same as the rate of convergence for the parameter estimates. In this example, we consider a shift in the mean of a normal distribution, which also demonstrates that our results extend beyond the label shift setting.

\begin{example}[Gaussian mean shift]
\label{ex:normal-shift}
Suppose it is known that under $\mathbb{P}_\infty$, $X \sim N(0, 1)$, and under $\mathbb{P}_0$, $X \sim N(\mu, 1)$, with $\mu > 0$. Then, the likelihood ratio $\lambda$ is given by $\lambda(x) = \exp\{\mu x - \mu^2/2\}$. Furthermore, we have a training sample $X_1',...,X_m' \overset{iid}{\sim} N(\mu, 1)$, with which we estimate $\mu$ by $\widehat{\mu}_m = \frac{1}{m} \sum \limits_{i=1}^m X_i'$; the likelihood ratio estimate is then $\widehat{\lambda}_m(x) = \exp\{\widehat{\mu}_m x - \widehat{\mu}_m^2/2\}$ (this is similar to the procedures proposed in \cite{tartakovsky2012efficient}). Then, the conditions for Theorem \ref{thm:inf-norm-rate-conv} hold with $TV(f^i_\lambda, f^i_{\widehat{\lambda}_m}) \leq O_P(|\widehat{\mu}_m - \mu|)$. Therefore, an upper bound on the rate of convergence for the estimated detection procedure is $|\mathbb{E}_i(\widetilde{T}^x(A) | \widehat{\lambda}_m) - \mathbb{E}_i(T^x(A))| \leq O_P(|\widehat{\mu}_m - \mu|) = O_P(1/\sqrt{m})$. Full details and calculations can be found in the Supplementary Materials (\ref{sec:ex-1-details}).
\end{example}

We now examine the case where we detect label shift with classifier predictions. To make the example clear, we consider a linear discriminant analysis (LDA) classifier, for which the classifier scores and their distribution have a closed form.

\begin{example}[LDA]
\label{ex:lda-label-shift} Suppose that under $\mathbb{P}_\infty$, $X \sim \pi_\infty N(\bm{\mu}_1, \bm{\Sigma}) + (1 - \pi_\infty)N(\bm{\mu}_0, \bm{\Sigma})$, while under $\mathbb{P}_0$,  $X \sim \pi_0 N(\bm{\mu}_1, \bm{\Sigma}) + (1 - \pi_0)N(\bm{\mu}_0, \bm{\Sigma})$. Suppose that $\pi_\infty$ and $\pi_0$ are known, and we have a training sample $X_1',...,X_m' \overset{iid}{\sim} \mathbb{P}_\infty$ with which we construct estimates $\widehat{\bm{\mu}}_1, \widehat{\bm{\mu}}_0, \widehat{\bm{\Sigma}}$. Our likelihood ratio estimate $\widehat{\lambda}_{\mathcal{A}, m}$ is given by \eqref{eq:est-lr}, where $\mathcal{A}$ is given by
\begin{align}
\label{eq:lda-score}
\mathcal{A}(X_i) = \frac{\pi_\infty \mathrm{MVN}(X_i; \widehat{\bm{\mu_1}}, \widehat{\bm{\Sigma}})}{\pi_\infty \mathrm{MVN}(X_i; \widehat{\bm{\mu_1}}, \widehat{\bm{\Sigma}}) + (1 - \pi_\infty)\mathrm{MVN}(X_i; \widehat{\bm{\mu_0}}, \widehat{\bm{\Sigma}})},
\end{align}
where $\mathrm{MVN}(\cdot; \bm{\mu}, \bm{\Sigma})$ denotes the multivariate normal density with mean $\bm{\mu}$ and covariance matrix $\bm{\Sigma}$. The true likelihood ratio is similar, just replacing $\widehat{\bm{\mu}}_1, \widehat{\bm{\mu}}_0, \widehat{\bm{\Sigma}}$ with the true parameters. (A3) follows because $f^i_\lambda$ and $f^i_{\widehat{\lambda}}$ can be shown to be Lipschitz, while (A5) follows from the strong consistency of $\widehat{\bm{\mu}}_i$ and $\widehat{\bm{\Sigma}}$. The rate of convergence in Theorem \ref{thm:inf-norm-rate-conv} depends on the rate of convergence for $||\widehat{\bm{\Sigma}}^{-1} - \bm{\Sigma}^{-1}||_F$, where $||\cdot||_F$ denotes Frobenius norm. Details are provided in the Supplementary Materials (\ref{sec:ex-2-details}).
\end{example}

While the formal results in Theorem \ref{thm:inf-norm-rate-conv} and Corollary \ref{cor:binary-convergence} depend on convergence of $\widehat{\lambda}(X)$ to $\lambda(X)$, detection performance will depend on classifier performance even if $\widehat{\lambda}(X)$ does not converge. Here we demonstrate that the relationship between detection delay and classifier performance still holds for mis-specified classifiers, and that the expected $L_1$ distance from Theorem \ref{thm:inf-norm-rate-conv}, $\mathbb{E}_i(|\widehat{\lambda}(X) - \lambda(X)|)$, is a useful summary of classifier performance.

\begin{example}[LDA vs. QDA] Suppose we observe training data $(X_1', Y_1'),...,(X_m', Y_m') \in \mathbb{R}^{150}$, with $X|Y = 0 \sim
N(\bm{\mu_0}, \bm{\Sigma_0})$ and $X|Y = 1 \sim N(\bm{\mu_1}, \bm{\Sigma_1})$, with $\bm{\Sigma_0} \neq \bm{\Sigma_1}$. Construct linear discriminant analysis (LDA) and quadratic discriminant analysis (QDA) classifiers $\mathcal{A}_L$ and $\mathcal{A}_Q$. Using \eqref{eq:est-lr}, perform CUSUM changepoint detection with each classifier. Figure \ref{fig:lda-qda-plot} shows $\mathbb{E}_0(|\widehat{\lambda}(X) - \lambda(X)|)$ and $\mathbb{E}_0(\widetilde{T}(A))$ for the resulting LDA and QDA detection procedures; for each classifier, the threshold $A$ is chosen so that $\mathbb{E}_\infty(\widetilde{T}(A)) \approx 180$. For large $m$, $\mathcal{A}_Q$ outperforms $\mathcal{A}_L$ because the QDA assumptions are met whereas the LDA assumptions are violated, but for small $m$ we see that $\mathcal{A}_L$ does better due to the bias-variance trade-off. As shown in Figure \ref{fig:lda-qda-plot}, the LDA detection procedure does better than the QDA detection procedure exactly when $\mathcal{A}_L$ outperforms $\mathcal{A}_Q$.
\label{ex:lda-vs-qda}
\end{example}

\begin{figure}
\centering
\includegraphics[scale=0.5]{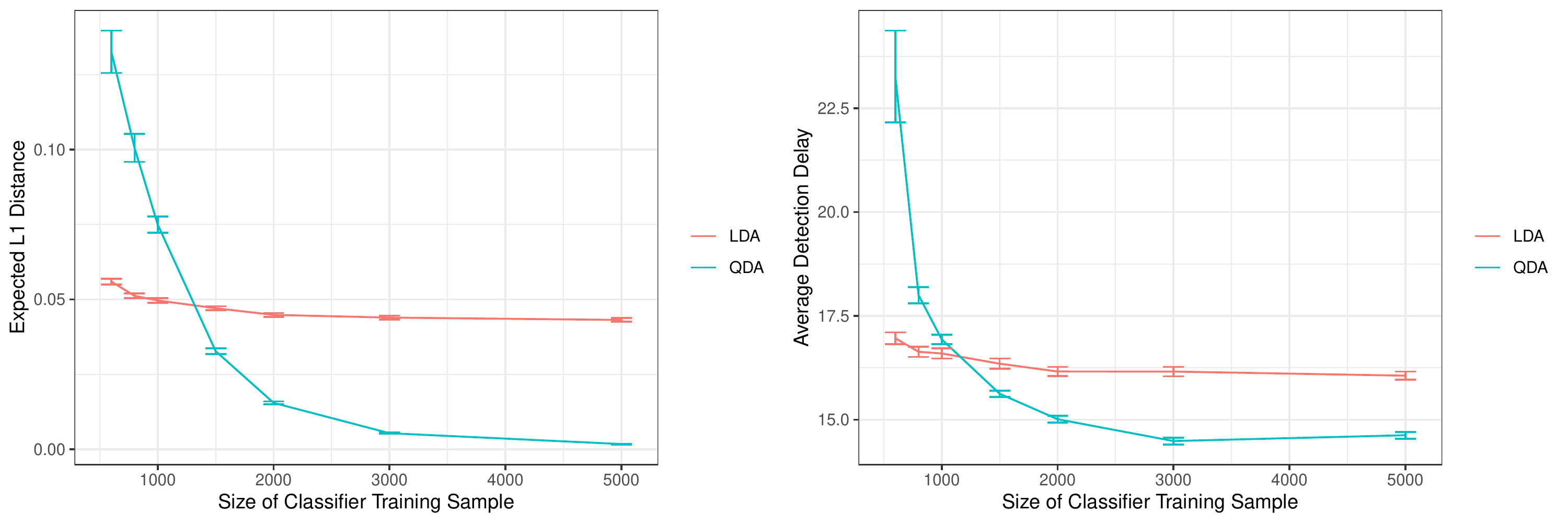}
\caption{Comparison between changepoint detection with LDA and QDA classifier scores, when LDA assumptions are violated but QDA assumptions are satisfied. \underline{Left}: Average L1 distance of the likelihood ratio estimate, as a function of the size of the classifier training sample. \underline{Right}: Each LDA and QDA classifier is used for changepoint detection with the classifier scores. This plot shows the average detection delay when the average ARL is approximately 180 (averaged across classifiers) as a function of the size of the classifier training sample.}
\label{fig:lda-qda-plot}
\end{figure}

In Example \ref{ex:lda-label-shift}, we used an LDA classifier to detect label shift in a mixture of normals with a common covariance matrix. Here we consider a similar procedure which uses nonparametric density estimation rather than an assumption of Gaussianity (see, e.g. \cite{cipolli2017computationally, cipolli2019supervised}).

\begin{example}[Density Estimation Classifier] Suppose that $\pi_\infty$ and $\pi_0$ are known, and we have a training sample $X_1',...,X_m' \overset{iid}{\sim} \mathbb{P}_\infty$ with which we construct density estimates $\widehat{f}_{\infty, X|Y=0}$ and $\widehat{f}_{\infty, X|Y=1}$. Our classifier $\mathcal{A}$ is given by 
\begin{align}
\mathcal{A}(x) = \frac{\pi_\infty \widehat{f}_{\infty, X|Y=1}(x)}{\pi_\infty \widehat{f}_{\infty, X|Y=1}(x) + (1 - \pi_\infty) \widehat{f}_{\infty, X|Y=0}(x)},
\end{align}
and our likelihood ratio estimate is given by \eqref{eq:est-lr}.
For $c > 0$, let $\mathcal{S}_c = \{x: f_{\infty, X}(x) > c\}$. Clearly, $\lim \limits_{c \to 0} \mathbb{P}_i(X \in \mathcal{S}_c) = 1$. Furthermore, with appropriate choice of kernel and bandwidth, $||\widehat{f}_{\infty, X|Y=1} - f_{\infty, X|Y=1}||_\infty, ||\widehat{f}_{\infty, X|Y=0} - f_{\infty, X|Y=0}||_\infty \overset{p}{\to} 0$ (see, e.g., \cite{gine2002rates}). Thus for each $c$, $\sup \limits_{\substack{x \in \mathcal{S}_c \\ \pi_0 \in \Pi_0}} | \widehat{\lambda}_{\pi_0, \mathcal{A}, m}(x) - \lambda_{\pi_0}(x)| \overset{p}{\to} 0$.
\end{example}

\section{Simulation studies}
\label{sec:simulations}

We investigate the empirical performance of the classifier-based label shift detection procedure described in Section \ref{sec:proposed-method}, with the likelihood ratio estimate in \eqref{eq:est-lr}. Our likelihood ratio estimate depends on a classifier, and for simplicity we will use an LDA classifier, since it is easy to control whether the LDA assumptions are satisfied, as in Example \ref{ex:lda-vs-qda}. For comparison, we consider several other detection procedures, which represent different approaches to changepoint detection. These procedures are summarized below and in Table \ref{tab:sim-assumptions}. Because the proposed procedure from Section \ref{sec:proposed-method} is designed specifically for the classifier label shift setting, it leverages more information than the other nonparametric detection procedures. In particular, as summarized in Table \ref{tab:sim-assumptions}, estimating the likelihood ratio with \eqref{eq:est-lr} assumes that the label shift assumption holds, and the classifier $\mathcal{A}(\cdot)$ performs well. Through simulations, we show that detection with \eqref{eq:est-lr} outperforms the other nonparametric procedures when these assumptions are met, and can still perform well when the assumptions are violated. While we use a simple setting for simulations, in Section \ref{sec:application} we apply the same methods to detect a change in dengue prevalence using the data and classifier from \cite{tuan2015sensitivity}, with similar results to our simulations in this section.

We compare the following methods:
\begin{description}
\item[Classifier-based CUSUM] This is the nonparametric method proposed in Section \ref{sec:proposed-method}, with likelihood ratio estimate \eqref{eq:est-lr}. For the purposes of simulations, $\mathcal{A}$ in \eqref{eq:est-lr} is an LDA classifier. Here we use a CUSUM procedure, so $\Psi(r) = \max\{1,r\}$.

\item[Optimal CUSUM] The optimal CUSUM procedure \citep{page1954continuous} uses the true likelihood ratio, and can be implemented when the true likelihood ratio is known. 

\item[uLSIF CUSUM] uLSIF \citep{kanamori2009least} is a nonparametric method for estimating the likelihood ratio, by maximizing an empirical divergence. As described in Section \ref{sec:proposed-method}, uLSIF can be used with training data under the label shift assumption by re-weighting or re-sampling training points, but it does not exploit the label shift structure of the likelihood ratio. A variety of similar density ratio estimation approaches exist, including KLIEP and kernel mean matching \citep{sugiyama2008direct, gretton2009covariate, kanamori2012statistical}, and we take uLSIF as a representative. Here we used the \texttt{densratio} package \citep{densratio} to implement uLSIF, and employ the resulting estimate in a CUSUM procedure.

\item[CPM] \cite{ackerman2020sequential} perform nonparametric label shift detection using the CPM framework described in \cite{ross2011nonparametric} and \cite{cpm}. The CPM framework detects changes in a sequence of univariate data using repeated nonparametric tests; \cite{ackerman2020sequential} applied repeated Cramer--von-Mises tests to a sequence of cosine divergences calculated between new data and training data. We evaluate CPM applied to both the \textbf{classifier} predictions and the cosine \textbf{divergences} used by \cite{ackerman2020sequential}. CPM stopping times are calculated with the \texttt{cpm} package \citep{cpm}.

\item[kNN] \cite{chen2019sequential} and \cite{chu2018sequential} propose a sequential graph-based $k$-nearest neighbors (kNN) detection procedure, based on repeated nearest-neighbor two-sample tests in a sliding window. Note that while the kNN approach uses training data, only a fixed window of data is considered. Similar to some parameters in \cite{chu2018sequential}, we set the window size to 200 and the number of nearest neighbors to $k=5$. Stopping times are calculated with the \texttt{gStream} package \citep{gstream}.
\end{description}

\begin{table}
\caption{\label{tab:sim-assumptions} Comparison of the information used by each changepoint detection procedure considered in simulations. CPM and kNN are more general than the classifier CUSUM procedure from Section \ref{sec:proposed-method}, but as a result they leverage less information. If the label shift assumption holds and the classifier performs well, we expect the classifier-based CUSUM method to outperform these more general procedures.}
\centering
\begin{tabular}{|c|c|l|c|c|c|c|}
\hline
\begin{tabular}[c]{@{}c@{}}Information \\ leveraged\end{tabular} & \begin{tabular}[c]{@{}c@{}}Classifier\\ CUSUM\end{tabular} & \multicolumn{1}{c|}{\begin{tabular}[c]{@{}c@{}}Optimal\\ CUSUM\end{tabular}} & \begin{tabular}[c]{@{}c@{}}uLSIF\\ CUSUM\end{tabular} & \begin{tabular}[c]{@{}c@{}}CPM\\ (classifier)\end{tabular} & \begin{tabular}[c]{@{}c@{}}CPM\\ (divergence)\end{tabular} & kNN          \\ \hline
True labels                                                      &                                                            & \multicolumn{1}{c|}{$\checkmark$}                                            &                                                       & \multicolumn{1}{l|}{}                                      &                                                            &              \\ \hline
Label shift                                                      & $\checkmark$                                               &                                                                              & $\checkmark$                                          & \multicolumn{1}{l|}{}                                      &                                                            &              \\ \hline
Good classifier                                                  & $\checkmark$                                               &                                                                              &                                                       & $\checkmark$                                               &                                                            &              \\ \hline
Training data                                                    & $\checkmark$                                               &                                                                              & $\checkmark$                                          & $\checkmark$                                               & $\checkmark$                                               & $\checkmark$ \\ \hline
\end{tabular}
\end{table}

\subsection{Metrics} 

Performance of each detection procedure is measured by detection delay, calculated as $\mathbb{E}_0[T]$ (for CUSUM procedures, this corresponds to Lorden's \citep{lorden1971procedures} detection delay). As is standard, we compare detection delays with each method calibrated to have the same average run length $\mathbb{E}_\infty[T]$. Here we use $\mathbb{E}_\infty[T] = 500$, which is a common value in the sequential detection literature. Expected stopping times are estimated via Monte Carlo simulation.

\subsection{Scenarios} 

Under the label shift assumption, the classifier-based CUSUM procedure uses classifier predictions $\mathcal{A}(X_i)$ to estimate the likelihood ratio. To compare performance of the different detection procedures, we use two different simulation scenarios. In the first scenario, we change the training sample size and the performance of the classifier (by changing the distribution of the data $X_i$ and violating LDA assumptions). In the second scenario, we change the performance of the classifier and the suitability of the label shift assumption.

\begin{description}
\item[Scenario 1:] Pre-change data is generated as $X \sim \pi_\infty N(\bm{\mu_1}, \bm{\Sigma_1}) + (1 - \pi_\infty) N(\bm{\mu_0}, \bm{\Sigma_0})$, and post-change data is generated as $X \sim \pi_0 N(\bm{\mu_1}, \bm{\Sigma_1}) + (1 - \pi_0) N(\bm{\mu_0}, \bm{\Sigma_0})$. In all simulations, $\pi_\infty = 0.4$, $\pi_0 = 0.7$, $\bm{\mu_0} = [0, 0]$, $\bm{\mu_1} = [1.5, 1.5]$, and $\bm{\Sigma_0} = \bm{I}$. Training data $(X_1', Y_1'),...,(X_m', Y_m')$ is simulated from the pre-change distribution, and used to train the LDA classifier, estimate the uLSIF likelihood ratio, and startup the CPM and kNN detection statistics. We consider $m \in \{200, 1000, 5000\}$, and $\bm{\Sigma_1} \in \left\lbrace \bm{I}, \begin{bmatrix}
2 & 0.1\\
0.1 & 2
\end{bmatrix}, \begin{bmatrix}
4 & 0.5\\
0.5 & 4
\end{bmatrix} \right\rbrace$.

\item[Scenario 2:] Pre-change data is generated as $X \sim \pi_\infty N(\bm{\mu_{\infty, 1}}, \bm{\Sigma_1}) + (1 - \pi_\infty) N(\bm{\mu_{\infty, 0}}, \bm{\Sigma_0})$, and post-change data is generated as $X \sim \pi_0 N(\bm{\mu_{0,1}}, \bm{\Sigma_1}) + (1 - \pi_0) N(\bm{\mu_{0,0}}, \bm{\Sigma_0})$. In all simulations, $\pi_\infty = 0.4$, $\pi_0 = 0.7$, $\bm{\mu_{\infty,0}} = [0, 0]$, $\bm{\mu_{\infty, 1}} = [1.5, 1.5]$, and $\bm{\Sigma_0} = \bm{I}$. Training data $(X_1', Y_1'),...,(X_{1000}', Y_{1000}')$ is simulated from the pre-change distribution, and used to train the LDA classifier, estimate the uLSIF likelihood ratio, and startup the CPM and kNN detection statistics. We consider $\bm{\Sigma_1} \in \left\lbrace \bm{I}, \begin{bmatrix}
2 & 0.1\\
0.1 & 2
\end{bmatrix}, \begin{bmatrix}
4 & 0.5\\
0.5 & 4
\end{bmatrix} \right\rbrace$ and the following pairs for $\bm{\mu_{0,0}}$ and $\bm{\mu_{0,1}}$: $\bm{\mu_{0,0}} = [0.5, 0.5]$ and $\bm{\mu_{0,1}} = [1, 1]$; $\bm{\mu_{0,0}} = [0.75, 0.75]$ and $\bm{\mu_{0,1}} = [0.75, 0.75]$; and $\bm{\mu_{0,0}} = [1, 1]$ and $\bm{\mu_{0,1}} = [0.5, 0.5]$.
\end{description}

\subsection{Results}

Table \ref{tab:sim-results-1} shows the results for Scenario 1, when the label shift assumption holds. We can see that when the LDA assumptions are met (specifically $\bm{\Sigma_1} = \bm{\Sigma_0} = \bm{I}$), LDA performs very close to the optimal CUSUM procedure, as we would predict from Example \ref{ex:lda-label-shift}. Performance of the LDA detection procedure relative to the optimal CUSUM procedure declines as the assumption that $\bm{\Sigma_1} = \bm{\Sigma_0}$ is violated, but is still better than the other nonparametric methods. This suggests that if the label shift assumption holds, the likelihood ratio estimate in \eqref{eq:est-lr} is a good choice for detecting the change. Detection with the uLSIF procedure improves with training sample size $m$, as it becomes easier to estimate the likelihood ratio function and variability in the likelihood ratio estimate decreases. CPM also performs better as the sample size increases, as training data is used to construct the detection statistic. While the kNN method makes no assumptions about the change or the distribution of data, the cost of this flexibility is a decrease in detection performance.

\begin{table}
\caption{\label{tab:sim-results-1} Simulation results for Scenario 1. Performance of each procedure is measured by detection delay, calculated as $\mathbb{E}_0[T]$. The estimated detection delay from Monte Carlo simulation is reported, with the standard error in parentheses. For the kNN procedure, a window of size 200 is used, so only 200 training points are considered. In the case of kNN, if a change is not detected within the sliding window, windows after time point 200 will consist of only post-change observations, so for computational purposes a fixed number of post-change observations is simulated and we report a lower bound on the detection delay.}
\centering
\begin{tabular}{|c|c|c|c|c|c|c|c|}
\hline
\multirow{2}{*}{$\bm{\Sigma_1}$} & \multirow{2}{*}{$m$} & \multicolumn{6}{c|}{Detection delay when $\mathbb{E}_\infty[T] \approx 500$}                                                                                                                                                                            \\ \cline{3-8} 
                            &                      & \begin{tabular}[c]{@{}c@{}}Classifier\\ CUSUM\end{tabular} & \begin{tabular}[c]{@{}c@{}}Optimal\\ CUSUM\end{tabular} & \begin{tabular}[c]{@{}c@{}}uLSIF\\ CUSUM\end{tabular} & \begin{tabular}[c]{@{}c@{}}CPM\\ (classifier)\end{tabular} & \begin{tabular}[c]{@{}c@{}}CPM\\ (divergence)\end{tabular} & kNN                                \\ \hline
\multirow{3}{*}{$\begin{bmatrix}
1 & 0\\
0 & 1
\end{bmatrix}$}          & 200                  & 28.8 (0.59)                                               & \multirow{3}{*}{29.0 (0.23)}                            & 33.8 (5.62)                                           & 46.9 (1.45)      & 51.5 (1.72)      & \multirow{3}{*}{$\geq 155$ (3.46)} \\ \cline{2-3} \cline{5-7}
                            & 1000                 & 28.8 (0.26)                                                &                                                         & 33.4 (1.51)                                           & 35.0 (0.81)      & 37.0 (0.85)      &                                    \\ \cline{2-3} \cline{5-7}
                            & 5000                 & 28.8 (0.10)                                                &                                                         & 31.5 (1.11)                                           & 32.6 (0.76)      & 33.7 (0.81)      &                                    \\ \hline
\multirow{3}{*}{$\begin{bmatrix}
2 & 0.1\\
0.1 & 2
\end{bmatrix}$}          & 200                  & 34.3 (0.96)                                                & \multirow{3}{*}{33.1 (0.28)}                            & 49.0 (11.4)                                           & 61.2 (2.24)      & 69.5 (2.67)      & \multirow{3}{*}{$\geq 168$ (3.58)} \\ \cline{2-3} \cline{5-7}
                            & 1000                 & 34.0 (0.40)                                                &                                                         & 43.3 (2.39)                                           & 42.4 (1.07)      & 46.6 (1.18)      &                                    \\ \cline{2-3} \cline{5-7}
                            & 5000                 & 34.0 (0.16)                                                &                                                         & 38.8 (2.44)                                           & 38.1 (0.94)      & 41.7 (1.03)      &                                    \\ \hline
\multirow{3}{*}{$\begin{bmatrix}
4 & 0.5\\
0.5 & 4
\end{bmatrix}$}          & 200                  & 41.9 (1.44)                                                & \multirow{3}{*}{33.4 (0.29)}                            & 164 (47.3)                                            & 74.6 (2.82)      & 94.8 (3.56)      & \multirow{3}{*}{$\geq 176$ (3.59)} \\ \cline{2-3} \cline{5-7}
                            & 1000                 & 41.8 (0.60)                                                &                                                         & 73.3 (7.05)                                           & 52.1 (1.46)      & 62.4 (1.74)      &                                    \\ \cline{2-3} \cline{5-7}
                            & 5000                 & 41.9 (0.25)                                                &                                                         & 54.2 (5.19)                                           & 51.2 (1.36)      & 58.5 (1.55)      &                                    \\ \hline
\end{tabular}
\end{table}

Table \ref{tab:sim-results-2} shows the results for Scenario 2, when the label shift assumption is violated. When the label shift assumption is approximately true ($\bm{\mu_{0,0}} = [0.5, 0.5]$ and $\bm{\mu_{0,1}} = [1, 1]$), we can see that LDA detection is comparable to uLSIF and CPM. However, the LDA procedure is more sensitive to large departures from the label shift assumption, for which methods with fewer assumptions perform better. Overall, CPM with classifier predictions performs well, as the classifier predictions are a useful summary of the data even when label shift doesn't hold.

\begin{table}
\caption{\label{tab:sim-results-2} Simulation results for Scenario 2. Performance of each procedure is measured by detection delay, calculated as $\mathbb{E}_0[T]$. The estimated detection delay from Monte Carlo simulation is reported, with the standard error in parentheses. For the kNN procedure, a window of size $200$ is used, so only 200 training points are considered. In the case of kNN, if a change is not detected within the sliding window, windows after time point 200 will consist of only post-change observations, so for computational purposes a fixed number of post-change observations is simulated and we report a lower bound on the detection delay.}
\centering
\tiny
\begin{tabular}{|c|c|c|c|c|c|c|c|}
\hline
\multirow{2}{*}{$\bm{\Sigma_1}$} & \multirow{2}{*}{\begin{tabular}[c]{@{}c@{}}Post-change \\ distribution\end{tabular}} & \multicolumn{6}{c|}{Detection delay when $\mathbb{E}_\infty[T] \approx 500$}                                                                                                                                                                                                                                               \\ \cline{3-8} 
                                 &                                                                                      & \begin{tabular}[c]{@{}c@{}}Classifier\\ CUSUM\end{tabular} & \begin{tabular}[c]{@{}c@{}}Optimal\\ CUSUM\end{tabular} & \begin{tabular}[c]{@{}c@{}}uLSIF\\ CUSUM\end{tabular} & \begin{tabular}[c]{@{}c@{}}CPM\\ (classifier)\end{tabular} & \begin{tabular}[c]{@{}c@{}}CPM\\ (divergence)\end{tabular} & kNN               \\ \hline
\multirow{3}{*}{$\begin{bmatrix}
1 & 0\\
0 & 1
\end{bmatrix}$}               & $\!\begin{aligned} \bm{\mu_{0,0}} &= [0.5, 0.5] \\ 
               \bm{\mu_{0,1}} &= [1, 1] \end{aligned}$                                                                                    & 70.2 (1.68)                                                & 24.1 (0.17)                                             & 55.8 (4.49)                                           & 64.4 (1.27)                                                & 66.4 (1.40)                                                & $\geq 130$ (3.22) \\ \cline{2-8} 
                                 & $\!\begin{aligned} \bm{\mu_{0,0}} &= [0.75, 0.75] \\ 
               \bm{\mu_{0,1}} &= [0.75, 0.75] \end{aligned}$                                                                                    & 184 (7.49)                                                 & 22.9 (0.16)                                             & 117 (14.6)                                            & 92.1 (1.64)                                                & 96.0 (1.78)                                                & $\geq 131$ (3.30) \\ \cline{2-8} 
                                 & $\!\begin{aligned} \bm{\mu_{0,0}} &= [1, 1] \\ 
               \bm{\mu_{0,1}} &= [0.5, 0.5] \end{aligned}$                                                                                   & 541 (17.4)                                                 & 28.9 (0.21)                                             & 282 (3.23)                                            & 177 (3.18)                                                 & 193 (3.55)                                                 & $\geq 161$ (3.57) \\ \hline
\multirow{3}{*}{$\begin{bmatrix}
2 & 0.1\\
0.1 & 2
\end{bmatrix}$}               & $\!\begin{aligned} \bm{\mu_{0,0}} &= [0.5, 0.5] \\ 
               \bm{\mu_{0,1}} &= [1, 1] \end{aligned}$                                                                                    & 70.9 (1.66)                                                & 34.3 (0.27)                                             & 67.9 (6.17)                                           & 74.6 (1.78)                                                & 79.4 (1.92)                                                & $\geq 173$ (3.60) \\ \cline{2-8} 
                                 & $\!\begin{aligned} \bm{\mu_{0,0}} &= [0.75, 0.75] \\ 
               \bm{\mu_{0,1}} &= [0.75, 0.75] \end{aligned}$                                                                                    & 123 (3.82)                                                 & 31.4 (0.26)                                             & 101 (11.5)                                            & 109 (2.69)                                                 & 120 (3.03)                                                 & $\geq 177$ (3.59) \\ \cline{2-8} 
                                 & $\!\begin{aligned} \bm{\mu_{0,0}} &= [1, 1] \\ 
               \bm{\mu_{0,1}} &= [0.5, 0.5] \end{aligned}$                                                                                    & 225 (8.56)                                                 & 30.0 (0.24)                                             & 166 (20.4)                                            & 196 (5.42)                                                 & 217 (6.15)                                                 & $\geq 177$ (3.60) \\ \hline
\multirow{3}{*}{$\begin{bmatrix}
4 & 0.5\\
0.5 & 4
\end{bmatrix}$}               & $\!\begin{aligned} \bm{\mu_{0,0}} &= [0.5, 0.5] \\ 
               \bm{\mu_{0,1}} &= [1, 1] \end{aligned}$                                                                                    & 73.8 (1.67)                                                & 29.8 (0.24)                                             & 87.6 (10.1)                                           & 77.0 (1.98)                                                & 88.7 (2.49)                                                & $\geq 170$ (3.57) \\ \cline{2-8} 
                                 & $\!\begin{aligned} \bm{\mu_{0,0}} &= [0.75, 0.75] \\ 
               \bm{\mu_{0,1}} &= [0.75, 0.75] \end{aligned}$                                                                                    & 103 (2.78)                                                 & 22.9 (0.19)                                             & 101 (11.9)                                            & 88.9 (2.55)                                                & 101 (3.08)                                                 & $\geq 152$ (3.50) \\ \cline{2-8} 
                                 & $\!\begin{aligned} \bm{\mu_{0,0}} &= [1, 1] \\ 
               \bm{\mu_{0,1}} &= [0.5, 0.5] \end{aligned}$                                                                                    & 146 (4.58)                                                 & 18.1 (0.15)                                             & 120 (15.3)                                            & 107 (3.12)                                                 & 122 (3.75)                                                 & $\geq 120$ (3.20) \\ \hline
\end{tabular}
\end{table}

\section{Detecting a change in dengue prevalence}
\label{sec:application}

We apply our changepoint detection procedure to the problem of detecting a change in the prevalence of dengue, using data and classifier predictions from the work of \cite{tuan2015sensitivity}. As the prevalence of dengue changes, but the symptoms are expected to stay the same, the label shift assumption is appropriate for this change.

\subsection{Data}

Data comes from \cite{tuan2015sensitivity}, who collected information on 5720 febrile patients aged 15 or younger in three Vietnamese hospitals. Of these patients, 30\% had dengue. The authors recorded their true dengue status (using a gold-standard test), the results of an NS1 rapid antigen test, and a variety of physical measurements for classification with a logistic regression classifier.

\subsection{Classifier} 

We use 1000 patients as training data for the classifier, and save the rest for evaluating our classifier and estimating changepoint detection performance. With the training set, we construct a logistic GAM classifier to predict true dengue status with the following covariates: vomiting (yes/no), skin bleeding (yes/no), BMI, age, temperature, white blood cell count, hematocrit, and platelet count. As in \cite{tuan2015sensitivity}, the ROC curve has an AUC of approximately 0.8.

\subsection{Scenarios} 

To assess change detection, we simulate a change in the prevalence of dengue by resampling the 4720 patients not used for training. As the group of patients in the study aims to represent the population of patients who would be tested for dengue, we take the sample proportion of 30\% as our baseline dengue prevalence among patients who would be tested. The degree of change in this prevalence, when an outbreak occurs, depends on the magnitude of the outbreak and the baseline prevalence in the population. Magnitude of change varies; for example, Hanoi, Vietnam saw roughly a five-fold increase in 2009 and 2015 \citep{cuong2011quantifying, cheng2020heatwaves}, while Kaohsiung City, Taiwan saw a 15-fold increase in 2014 \citep{hsu2017trend}. Baseline prevalence in the full population varies depending on location -- for example, \cite{wiwanitkit2006observation} shows approximately 1 in 1 million for certain areas of Thailand, whereas \cite{hsu2017trend} show roughly 1 in 10000 on average in Taiwan. For Vietnam, \cite{cuong2011quantifying} report roughly 1 in 10000 to 1 in 1000 in Hanoi, with a peak of 384 per 100000 in 2009. For our purposes, we consider two label shift changes in prevalence:
\begin{description}
\item[Abrupt change:] We simulate an abrupt 5-fold increase, and take the baseline prevalence in the population to be roughly 1 in 10000. Applying Bayes rule, this gives a post-change prevalence of about 68\% in our study population, and so we simulate a change from 30\% to 68\% and assess our ability to detect this shift.

\item[Gradual change:] When the change occurs, prevalence increases gradually, rather than abruptly. Here, prevalence in the study population changes smoothly from 30\% to 68\% over the course of 100 observations.
\end{description}

\subsection{Methods for comparison in dengue setting} 

We compare the methods from Section \ref{sec:simulations} to detect the change in dengue prevalence. The classifier CUSUM detection procedure is implemented using \eqref{eq:est-lr} with $\mathcal{A}(X)$ the predicted probabilities from the dengue classifier described above. We also compare CUSUM with binarized predictions, using both a threshold of 0.5 and the threshold, 0.33, which maximizes sensitivity + specificity. The optimal CUSUM procedure uses the true dengue status, which is observable if gold-standard tests are available, and we also include CUSUM with binary predictions from the NS1 rapid antigen tests, which again may not be available. The rapid test has a specificity of approximately 99\% and a sensitivity of 70\% \citep{tuan2015sensitivity}, compared with a specificity and sensitivity of 82\% and 70\% for the binarized classifier at threshold 0.33. As in Section \ref{sec:simulations}, we also compare CPM using the classifier predicted probabilities, and CPM with divergences. uLSIF was considered but failed to consistently estimate the likelihood ratio, while kNN was not considered because it performed worse than the other methods in Section \ref{sec:simulations}. Finally, as the post-change parameter is typically unknown, we include the mixing procedure described in \eqref{eq:rtw}. We use $\Pi_0 = [0.6, 0.8]$, which corresponds to a 3.5-fold to 9-fold increase in prevalence.

For the abrupt change scenario, all methods are compared. For the gradual change, we compare the mixture CUSUM procedure to CPM with classifier predictions, as these two methods perform well at detecting an abrupt change and do not require knowledge of the post-change parameter, and we include optimal CUSUM for reference.

\begin{table}
\caption{\label{tab:dengue-results} Comparison of method performance for detecting an abrupt change in dengue prevalence. Performance of each procedure is assessed by average detection delay ($\mathbb{E}_0[T]$), calculated at three different values of average run length ($\mathbb{E}_\infty[T]$). The estimated detection delay from Monte Carlo simulation is reported, with the standard error in parentheses.}
\centering
\begin{tabular}{|c|c|c|c|}
\hline
\multirow{2}{*}{Method}                                                              & \multicolumn{3}{c|}{Detection delay for three values of $\mathbb{E}_\infty[T]$} \\ \cline{2-4} 
                                                                                     & $\mathbb{E}_\infty[T] = 500$           & $\mathbb{E}_\infty[T] = 700$          & $\mathbb{E}_\infty[T] = 1000$         \\ \hline
Optimal CUSUM                                                                        & 11.73 (0.06)  & 12.56 (0.06) & 13.62 (0.07) \\ \hline
Rapid test CUSUM                                                                     & 19.46 (0.15)  & 23.21 (0.20) & 24.66 (0.20) \\ \hline
Mixture CUSUM                                                                        & 25.56 (0.68)  & 27.52 (0.70) & 30.04 (0.75) \\ \hline
\begin{tabular}[c]{@{}c@{}}Classifier CUSUM\\ (predicted probability)\end{tabular}   & 26.28 (0.16)  & 29.06 (0.17) & 31.67 (0.18) \\ \hline
\begin{tabular}[c]{@{}c@{}}Classifier CUSM\\ (binary, threshold = 0.33)\end{tabular} & 30.54 (0.22)  & 33.58 (0.23) & 37.54 (0.26) \\ \hline
CPM (classifier)                                                                     & 36.2 (0.28)   & 40.4 (0.30)  & 44.6 (0.32)  \\ \hline
\begin{tabular}[c]{@{}c@{}}Classifier CUSUM\\ (binary, threshold = 0.5)\end{tabular} & 41.22 (0.37)  & 49.72 (0.45) & 56.04 (0.51) \\ \hline
CPM (divergence)                                                                     & 63.0 (0.54)   & 72.3 (0.60)  & 81.7 (0.67)  \\ \hline
uLSIF CUSUM                                                                          & 1295 (88)     & 1745 (111)   & 2305 (141)   \\ \hline
\end{tabular}
\end{table}

\subsection{Results for dengue example}

\subsubsection{Abrupt change} 

Figure \ref{fig:dengue-detect} and Table \ref{tab:dengue-results} show the relationship between $\mathbb{E}_\infty[T]$ (average time to false alarm) and $\mathbb{E}_0[T]$ (average detection delay) for each method (uLSIF is not shown in Figure \ref{fig:dengue-detect} because the detection delays are too large). As expected from Corollary \ref{cor:binary-convergence}, the true dengue status and the rapid antigen test give the best detection performance. The predicted probabilities outperform the binarized predictions, as binarization throws away information on the likelihood ratio. The two binarized predictions are close, but the optimal threshold -- which maximizes sensitivity + specificity -- performs better, as predicted by Corollary \ref{cor:binary-convergence}. Mixture CUSUM and CUSUM with the predicted probabilities perform equally well, likely because all $\pi_0 \in \Pi_0 = [0.6, 0.8]$ provide similar results. While CPM performs worse than CUSUM with predicted probabilities, it still provides a competitive alternative that requires no assumptions on the post-change prevalence. uLSIF has difficulty estimating the likelihood ratio, and performs substantially worse than the other methods.

\begin{figure}
\centering
\includegraphics[scale=0.55]{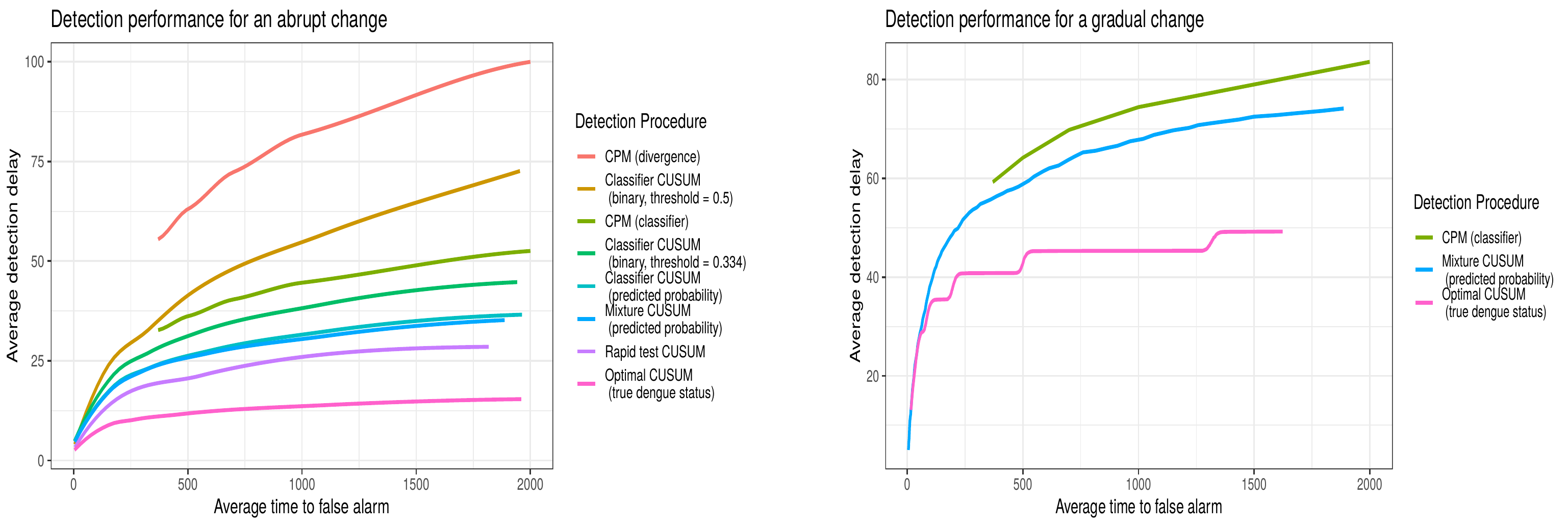}
\caption{\underline{Left}: Comparison of detection performance for CUSUM procedures using different detection procedures, for a change in dengue prevalence from $\pi_\infty = 0.3$ to $\pi_0 = 0.68$. For ease, the method labels for the plot are displayed in descending order of detection delay. \underline{Right}: Comparison of detection performance when $\pi_0$ changes gradually from $0.3$ to $0.68$.}
\label{fig:dengue-detect}
\end{figure}

\subsubsection{Gradual change} 

Figure \ref{fig:dengue-detect} shows the relationship between  $\mathbb{E}_\infty[T]$ and $\mathbb{E}_0[T]$ for each method. Detection delays are longer for all methods under gradual change than abrupt change, because the magnitude of change is initially smaller. However, each method can raise an alarm reasonably quickly. This is valuable because real changes in prevalence are expected to be continuous, rather than an abrupt switch from one prevalence to another. While the classic CUSUM procedure, and the nonparametric methods discussed in this paper, are designed to detect an abrupt change, Figure \ref{fig:dengue-detect} demonstrates that these methods are sensitive to other changes too.

\section{Discussion}

When a classifier is applied sequentially over time, it is important to detect any change in the distribution of classification data. First, distributional shifts can affect the validity of classifier predictions, and second, a change in distribution may suggest a problem like a disease outbreak. In this paper, we consider procedures for detecting label shift, which can occur when the prevalence of a disease changes over time, but the symptoms of the disease remain the same.

As we focus on detecting changes in classification data, it is natural to use the classifier predictions in our detection procedure. Here we propose a simple, nonparametric sequential changepoint detection method that uses the classifier predictions to approximate the true likelihood ratio \eqref{eq:est-lr}. Our procedure requires no additional estimation or training, assuming only that a reasonable value of the post-change prevalence $\pi_0$ can be specified. Furthermore, when this post-change parameter is unknown, we combine our nonparametric procedure with Lai's mixture CUSUM approach \citep{lai1998information}, and mix over the unknown prevalence. 

Performance of the detection procedure then depends directly on classifier performance. To demonstrate this, in Section \ref{sec:theory} we introduce new convergence results for nonparametric sequential detection procedures with likelihood ratio estimates. Through simulations in Section \ref{sec:simulations}, we illustrate that our proposed detection procedure outperforms other nonparametric methods when the label shift assumption holds, and still achieves comparable performance when the the label shift assumption is violated. The same holds true when these methods are applied to real dengue classification data in Section \ref{sec:application}, in which we apply the classifier described in \cite{tuan2015sensitivity} to detect a simulated dengue outbreak. First, we see that improved classifier performance results in improved detection performance -- if the gold standard dengue test is unavailable, only the NS1 rapid antigen test (which has better specificity than the classifier from \cite{tuan2015sensitivity}) outperforms our proposed procedure. Second, other nonparametric procedures respond more slowly to the outbreak, because they leverage less information about a change in prevalence.

While the label shift assumption might not hold exactly in real data, \cite{rabanser2019failing} found that testing for label shift is still a useful way for finding more general changes in distribution. This is supported by our simulation results, in which our label shift detection procedure still performs well under mild violations of the label shift assumption.

\section*{Acknowledgments}

Both authors were partially supported by NSF DMS1613202. The authors gratefully acknowledge Samuel Ackerman, Chris Genovese, Aaditya Ramdas, Alessandro Rinaldo, and Zack Lipton for helpful discussions and feedback.

\bibliographystyle{apalike}
\bibliography{references}

\newpage

\appendix

\section{Supplementary Materials}

\subsection{Data and code}

Full code for the data analysis and simulations presented in this paper is available at \url{https://github.com/ciaran-evans/label-shift-detection}. The data used in the dengue case study was made publicly available by \cite{tuan2015sensitivity}, and a copy is provided in the repository with the code.

\subsection{Proof of Theorem \ref{thm:inf-norm-rate-conv}}
\label{sec:theorem1-proof}

\begin{proof}
Following \citet{tartakovsky2009design}, we can express the desired expectations as solutions to Fredholm integral equations of the second kind. In particular, let $v_1(x) = \mathbb{E}_i(T^x(A))$, and $v_2(x) = \mathbb{E}_i(\widetilde{T}^x(A) | \widehat{\lambda}_m)$. Then, 
\begin{align}
v_1(x) &= 1 + \int \limits_0^A v_1(y) k(x, y) dy \\
v_2(x) &= 1 + \int \limits_0^{A} v_2(y) \widetilde{k}(x,y) dy,
\end{align}
where $k(x,y) = \frac{\partial}{\partial y} \mathbb{P}_i\left(\lambda(X) \leq \frac{y}{\Psi(x)}\right)$ \citep{tartakovsky2009design}, and $\widetilde{k}(x,y) = \frac{\partial}{\partial y} \mathbb{P}_i\left(\widehat{\lambda}_m(X) \leq \frac{y}{\Psi(x)} | \widehat{\lambda}_m \right)$. Using this representation, we show that $||v_1 - v_2||_\infty$ is small.

First, let $K$ and $\widetilde{K}$ be the integral operators defined by the kernels $k(x,y)$ and $\widetilde{k}(x,y)$. Then under assumption (A3), we have $K, \widetilde{K} : C([0,A]) \to C([0,A])$ are compact \citep{han2009theoretical}. Furthermore, $v_1 = (I - K)\bm{1}$ and $v_2 = (I - \widetilde{K})\bm{1}$, where $\bm{1}(x) \equiv 1$.

Next we show that $Kv = v$ and $\widetilde{K}v = v$ have only the trivial solution $v = 0$. Considering $K, \widetilde{K}$ on $L^2(0,A)$, we have $K, \widetilde{K} : L^2(0,A) \to L^2(0, A)$ because $k, \widetilde{k}$ are Hilbert-Schmidt kernel functions under (A3). Then consider the adjoint operators $K^*$ and $\widetilde{K}^*$. \cite{moustakides2011numerical} demonstrate that the maximal eigenvalue of $K^*$ is strictly less than 1 when $\lambda(X)$ is continuous, and therefore the same holds for $K$. So, $Kv = v$ has only the trivial solution $v = 0$, with the same result for $\widetilde{K}$.

Then by the Fredholm alternative theorem, $(I - K)^{-1}, (I - \widetilde{K})^{-1} : C([0, A]) \to C([0, A])$ are bijective and bounded \citep{han2009theoretical}. Furthermore, 
\begin{align}
||K - \widetilde{K}||_\infty \leq \sup \limits_{x \in [0,A]} \int \limits_0^A |k(x,y) - \widetilde{k}(x,y)| dy \leq \int \limits_0^\infty |f^i_\lambda(s) - f^i_{\widehat{\lambda}_m}(s)| ds = 2 TV(f^i_\lambda, f^i_{\widehat{\lambda}_m}).
\end{align}

By (A5), $TV(f^i_\lambda, f^i_{\widehat{\lambda}_m}) \overset{p}{\to} 0$, and so by Theorem 2.3.5 in \cite{han2009theoretical}, we have that
\begin{align}
|\mathbb{E}_i(\widetilde{T}^x(A) | \widehat{\lambda}_m) - \mathbb{E}_i(T^x(A))| \leq ||v_2 - v_1||_\infty \leq ||\widetilde{K}^{-1}||_\infty ||(K - \widetilde{K})v_1||_\infty \leq O_P(TV(f^i_\lambda, f^i_{\widehat{\lambda}_m})).
\end{align}

Finally, we can improve the bound in the case when $f^i_\lambda$ is bounded. Note that
\begin{align}
\begin{split}
\left\lvert \int \limits_0^A (k(x,y) - \widetilde{k}(x,y)) v_1(y) dy \right\rvert &\leq \left \lvert \int \limits_0^\infty f_{\Psi(x)\lambda}(y) \overline{v}_1(y)dy - \int \limits_0^\infty f_{\Psi(x)\widehat{\lambda}_m}(y) \overline{v}_1(y) dy \right\rvert \\ 
& \hspace{0.5cm} + v_1(A) \left\lvert \int \limits_A^\infty (f_{\Psi(x)\lambda}(y) - f_{\Psi(x)\widehat{\lambda}_m}(y)) dy \right\rvert,
\end{split}
\end{align}
where $\overline{v}_1(y) = v_1(y)$ for $y \leq A$, and $\overline{v}_1(y) = v_1(A)$ for $y \geq A$. If $f^i_\lambda$ is bounded, then $v_1$ is Lipschitz, and 
\begin{align}
\left \lvert \int \limits_0^\infty f_{\Psi(x)\lambda}(y) \overline{v}_1(y)dy - \int \limits_0^\infty f_{\Psi(x)\widehat{\lambda}_m}(y) \overline{v}_1(y) dy \right\rvert &\leq C \sup \limits_{||h||_L \leq 1} \left \lvert \int \limits_0^\infty f_{\Psi(x)\lambda}(y) h(y)dy - \int \limits_0^\infty f_{\Psi(x)\widehat{\lambda}_m}(y) h(y) dy \right\rvert \\
& \leq C \Psi(A) \mathbb{E}_i[|\widehat{\lambda}_m(X) - \lambda(X)|],
\end{align}
by Kantorovich-Rubinstein duality. Since $\left\lvert \int \limits_A^\infty (f_{\Psi(x)\lambda}(y) - f_{\Psi(x)\widehat{\lambda}_m}(y)) dy \right\rvert \leq ||F^i_\lambda - F^i_{\widehat{\lambda}_m}||_\infty$, this concludes the proof.
\end{proof}

\subsection{Proof of Corollary \ref{cor:binary-convergence}}
\label{sec:cor-1-proof}

\begin{proof}
The proof of Corollary \ref{cor:binary-convergence} is similar to the proof of Theorem \ref{thm:inf-norm-rate-conv}, but in a finite-dimensional space. Since $\pi_0$ and $\pi_\infty$ are known, and $\log(\pi_0/\pi_\infty)$ and $\log((1 - \pi_0)/(1 - \pi_\infty))$ are both rational, then $R_t^x$ and $\widetilde{R}_t^x$ are Markov chains on the same finite state space, where $x$ and the states are linear combinations of $\log(\pi_0/\pi_\infty)$ and $\log((1 - \pi_0)/(1 - \pi_\infty))$. Let $v_1$ denote the vector of expected stopping times when starting the optimal detection procedure in each state, and $v_2$ the corresponding vector for the estimated detection procedure. Then, $(I - K)v_1 = \bm{1}$ and $(I - \widetilde{K})v_2 = \bm{1}$, where $\bm{1}$ is the vector of all 1's, and $K$ and $\widetilde{K}$ are transition probability matrices for the Markov chain.

In particular, all elements of $K$ are either 0, $\mathbb{P}(Y_i = 0)$, or $\mathbb{P}(Y_i = 1)$. The corresponding elements of $\widetilde{K}$ are 
\begin{align}
\widetilde{K}_{ij} = \begin{cases}
0 & K_{ij} = 0 \\
\mathbb{P}(\mathcal{A}(X) = 0 | \mathcal{A}) & K_{ij} = \mathbb{P}(Y_i = 0) \\
\mathbb{P}(\mathcal{A}(X) = 1 | \mathcal{A}) & K_{ij} = \mathbb{P}(Y_i = 1).
\end{cases}
\end{align}
Therefore, $|\widetilde{K}_{ij} - K_{ij}| \leq O_P(\mathbb{P}_i(\mathcal{A}(X) = 0 | Y = 1, \mathcal{A}) + \mathbb{P}_i(\mathcal{A}(X) = 1 | Y = 0, \mathcal{A}))$, and as in Theorem \ref{thm:inf-norm-rate-conv} the proof follows again by applying Theorem 2.3.5 from \cite{han2009theoretical}.
\end{proof}

\subsection{Proof of Theorem \ref{thm:dens-est-mix}}
\label{sec:thm-2-proof}

\begin{proof}
For ease of notation, we drop the subscript $i \in \{0, \infty\}$ and the dependence on the classifier $\mathcal{A}$ from the expectations; our goal is to show $|\mathbb{E}(\widetilde{T}_w(A)) - \mathbb{E}(T_w(A))| \overset{p}{\to} 0$. Let $U_t(A) = \min\{T_w(A), t\}$ and $\widetilde{U}_t(A) = \min\{\widetilde{T}_w(A), t\}$. Then for any $t_0 > 0$,
\begin{align}
\begin{split}
|\mathbb{E}(\widetilde{T}_w(A)) - \mathbb{E}(T_w(A))| &\leq |\mathbb{E}(\widetilde{T}_w(A)) - \mathbb{E}(\widetilde{U}_{t_0}(A))| + |\mathbb{E}(\widetilde{U}_{t_0}(A)) - \mathbb{E}(U_{t_0}(A))| + \\
& \hspace{0.5cm} |\mathbb{E}(U_{t_0}(A)) - \mathbb{E}(T_w(A))|.
\end{split}
\end{align}

For $\varepsilon > 0$, let $\mathcal{C}_{t_0, \varepsilon} = \left\lbrace X_1,...,X_{t_0} : \sup \limits_{t \leq t_0} |\widetilde{R}_{t,w} - R_{t,w}| < \varepsilon | \mathcal{A}\right \rbrace$. Then 
\begin{align}
|\mathbb{E}(\widetilde{U}_{t_0}(A)) - \mathbb{E}(U_{t_0}(A))| \leq |\mathbb{E}(\widetilde{U}_{t_0}(A)) - \mathbb{E}(\widetilde{U}_{t_0}(A) | \mathcal{C}_{t_0, \varepsilon})| + |\mathbb{E}(\widetilde{U}_{t_0}(A) | \mathcal{C}_{t_0, \varepsilon}) - \mathbb{E}(U_{t_0}(A))|,
\end{align}
and $|\mathbb{E}(\widetilde{U}_{t_0}(A)) - \mathbb{E}(\widetilde{U}_{t_0}(A) | \mathcal{C}_{t_0, \varepsilon})| \leq 2 t_0 (1 - \mathbb{P}(\mathcal{C}_{t_0, \varepsilon}))$. Also, if $X_1,...,X_{t_0} \in \mathcal{C}_{t_0, \varepsilon}$ then $U_{t_0}(A - \varepsilon) \leq \widetilde{U}_{t_0}(A) \leq U_{t_0}(A + \varepsilon)$, and so
\begin{align}
\begin{split}
|\mathbb{E}(\widetilde{U}_{t_0}(A) | \mathcal{C}_{t_0, \varepsilon}) - \mathbb{E}(U_{t_0}(A))| &\leq |\mathbb{E}(U_{t_0}(A + \varepsilon)) - \mathbb{E}(T_w(A + \varepsilon))| + |\mathbb{E}(T_w(A + \varepsilon)) - \mathbb{E}(T_w(A - \varepsilon))| + \\
& \hspace{0.5cm} |\mathbb{E}(T_w(A - \varepsilon)) - \mathbb{E}(U_{t_0}(A - \varepsilon))|.
\end{split}
\end{align}

Now let $\eta > 0$. By continuity of $\mathbb{E}(T_w(A))$, there exists $\varepsilon > 0$ such that $|\mathbb{E}(T_w(A + \varepsilon)) - \mathbb{E}(T_w(A - \varepsilon))| < \eta/6$. Next, let $c_1 > 0$ and $0 < \delta_1 < \min\{\frac{a}{\pi_\infty}, \frac{1-b}{1-\pi_\infty}\}$. For $\pi_0 \in \Pi_0$, define $\lambda_{\pi_0}^{\min}$ by
\begin{align}
\lambda_{\pi_0}^{\min}(x) = \begin{cases}
\lambda_{\pi_0}(x) - \delta_1 & x \in \mathcal{S}_{c_1} \\
\min\{\frac{a}{\pi_\infty}, \frac{1-b}{1-\pi_\infty}\} & x \not \in \mathcal{S}_{c_1},
\end{cases}
\end{align}
and define $T_w^{\min}$ by replacing $\lambda_{\pi_0}$ with $\lambda_{\pi_0}^{\min}$ in \eqref{eq:rtw}, and where we choose $\delta_1$ and $c_1$ sufficiently small that $\mathbb{E}(T_w^{\min})$ is finite. Therefore, if $\sup \limits_{\substack{x \in \mathcal{S}_{c_1} \\ \pi_0 \in \Pi_0}} | \widehat{\lambda}_{\pi_0, \mathcal{A}, m}(x) - \lambda_{\pi_0}(x)| < \delta_1$, then $|\mathbb{E}(\widetilde{T}_w(A)) - \mathbb{E}(\widetilde{U}_{t_0}(A))| \leq |\mathbb{E}(T_w^{\min}(A)) - \mathbb{E}(U_{t_0}^{\min}(A))|$.

Now choose $t_0$ sufficiently large that
\begin{align}
\begin{split}
|\mathbb{E}(T_w^{\min}(A)) - \mathbb{E}(U_{t_0}^{\min}(A))| & < \eta/6 \\
|\mathbb{E}(U_{t_0}(A + \varepsilon)) - \mathbb{E}(T_w(A + \varepsilon))| &< \eta/6 \\
|\mathbb{E}(U_{t_0}(A - \varepsilon)) - \mathbb{E}(T_w(A - \varepsilon))| &< \eta/6 \\
|\mathbb{E}(U_{t_0}(A)) - \mathbb{E}(T_w(A))| &< \eta/6.
\end{split}
\end{align}

Finally, we just have to control $2t_0(1 - \mathbb{P}(\mathcal{C}_{t_0, \varepsilon}))$. If $\sup \limits_{\substack{x \in \mathcal{S}_{c} \\ \pi_0 \in \Pi_0}} | \widehat{\lambda}_{\pi_0, \mathcal{A}, m}(x) - \lambda_{\pi_0}(x)| < \delta$ and $X_1,...,X_{t_0} \in \mathcal{S}_c$, then 
\begin{align}
\begin{split}
\sup \limits_{t \leq t_0}|\widetilde{R}_{t,w} - R_{t,w}| &= \sup \limits_{t \leq t_0} \left\lvert \max \limits_{t - m_\alpha \leq k \leq t} \int \limits_{\Pi_0} \prod \limits_{i=k}^t \widehat{\lambda}_{\pi_0, \mathcal{A}, m}(X_i) w(\pi_0) d\pi_0 - \max \limits_{t - m_\alpha \leq k \leq t} \int \limits_{\Pi_0} \prod \limits_{i=k}^t \lambda_{\pi_0}(X_i) w(\pi_0) d\pi_0 \right\rvert \\
\\
& \leq \sup \limits_{t \leq t_0} \max \limits_{t - m_\alpha \leq k \leq t} \left\lvert \int \limits_{\Pi_0} \prod \limits_{i=k}^t \widehat{\lambda}_{\pi_0, \mathcal{A}, m}(X_i) w(\pi_0) d\pi_0 - \int \limits_{\Pi_0} \prod \limits_{i=k}^t \lambda_{\pi_0}(X_i) w(\pi_0) d\pi_0 \right\rvert \\
& \leq \sup \limits_{t \leq t_0} \int \limits_{\Pi_0} \left( \max \limits_{t - m_\alpha \leq k \leq t} \left\lvert \prod \limits_{i=k}^t \widehat{\lambda}_{\pi_0, \mathcal{A}, m}(X_i) - \prod \limits_{i=k}^t \lambda_{\pi_0}(X_i) \right\rvert \right) w(\pi_0) d\pi_0 \\
& \leq \max \left\lbrace (M + \delta)^{m_\alpha} - M^{m_\alpha}, \ M^{m_\alpha} - (M - \delta)^{m_\alpha}  \right\rbrace.
\end{split}
\end{align}
where $M = \max \left\lbrace \frac{1}{\pi_\infty} \ , \ \frac{1}{1 - \pi_\infty} \right\rbrace$. Choose $\delta_2 < \delta_1$ such that the right hand side is at most $\varepsilon$, and $c_2 < c_1$ such that $2t_0(1 - \mathbb{P}(X_i \in \mathcal{S}_{c_2})^{t_0}) < \eta/6$. Then, $\sup \limits_{\substack{x \in \mathcal{S}_{c_2} \\ \pi_0 \in \Pi_0}} | \widehat{\lambda}_{\pi_0, \mathcal{A}, m}(x) - \lambda_{\pi_0}(x)| < \delta_2$ implies that $|\mathbb{E}(\widetilde{T}_w(A)) - \mathbb{E}(T_w(A))| < \eta$, and so
\begin{align}
\mathbb{P}\left(|\mathbb{E}(\widetilde{T}_w(A)) - \mathbb{E}(T_w(A))| < \eta \right) \geq \mathbb{P}\left(\sup \limits_{\substack{x \in \mathcal{S}_{c_2} \\ \pi_0 \in \Pi_0}} | \widehat{\lambda}_{\pi_0, \mathcal{A}, m}(x) - \lambda_{\pi_0}(x)| < \delta_2 \right) \to 1.
\end{align}
As this works for all $\eta > 0$, then we conclude that $|\mathbb{E}(\widetilde{T}_w(A)) - \mathbb{E}(T_w(A))| \overset{p}{\to} 0$ as desired.
\end{proof}

\subsection{Details for Example \ref{ex:normal-shift}}
\label{sec:ex-1-details}

Under $\mathbb{P}_\infty$, $X \sim N(0, 1)$ and under $\mathbb{P}_0$, $X \sim N(\mu, 1)$. Therefore,
\begin{align}
\lambda(x) = \frac{\exp\{ -0.5 (x - \mu)^2 \}}{\exp\{-0.5 x^2\}} = \exp\left\lbrace \mu x - \frac{\mu^2}{2} \right\rbrace.
\end{align}
Let $\widehat{\mu}$ be an estimate of $\mu$, then $\widehat{\lambda}(x) = \exp\left\lbrace \widehat{\mu} x - \frac{\widehat{\mu}^2}{2} \right\rbrace$. 

Next, we need $f_\lambda^i$ and $f_{\widehat{\lambda}}^i$. Since $P(\widehat{\lambda}(X) \leq s) = P(X \leq \log(s)/\mu + \mu/2)$, then $f_\lambda^i(s) = f_X^i( \log(s)/\mu + \mu/2)/(\mu s)$. Likewise, $f_{\widehat{\lambda}}^i(s) = f_X^i( \log(s)/\widehat{\mu} + \widehat{\mu}/2)/(\widehat{\mu} s)$. Since the pre- and post-change distributions are normal, then
\begin{align}
f_\lambda^\infty(s) \propto \frac{1}{\mu s} \exp \left\lbrace -\frac{1}{2} \left( \frac{\log s}{\mu} + \frac{\mu}{2} \right)^2 \right\rbrace \hspace{0.5cm} f_{\widehat{\lambda}}^\infty(s) \propto \frac{1}{\widehat{\mu} s} \exp \left\lbrace -\frac{1}{2} \left( \frac{\log s}{\widehat{\mu}} + \frac{\widehat{\mu}}{2} \right)^2 \right\rbrace \\
f_\lambda^0(s) \propto \frac{1}{\mu s} \exp \left\lbrace -\frac{1}{2} \left( \frac{\log s}{\mu} - \frac{\mu}{2} \right)^2 \right\rbrace \hspace{0.5cm} f_{\widehat{\lambda}}^\infty(s) \propto \frac{1}{\widehat{\mu} s} \exp \left\lbrace -\frac{1}{2} \left( \frac{\log s}{\widehat{\mu}} + \frac{\widehat{\mu}}{2} - \mu \right)^2 \right\rbrace.
\end{align}

Clearly, (A2) is satisfied. To show that (A3) is satisfied, we prove that $f_\lambda^i$ and $f_{\widehat{\lambda}}^i$ are Lipschitz. We have
\begin{align}
\frac{d}{ds} f_{\lambda}^\infty(s) &= \frac{1}{(\mu s)^2} \exp \left\lbrace -\frac{1}{2} \left(\frac{\log s}{\mu} + \frac{\mu}{2} \right)^2 \right\rbrace \left( -\frac{\log s}{\mu} - \frac{3 \mu}{2} \right) \\
\frac{d}{ds} f_{\lambda}^0(s) &= \frac{1}{(\mu s)^2} \exp \left\lbrace -\frac{1}{2} \left(\frac{\log s}{\mu} - \frac{\mu}{2} \right)^2 \right\rbrace \left( -\frac{\log s}{\mu} - \frac{\mu}{2} \right).
\end{align}
Since both derivatives are bounded, $f_\lambda^\infty$ and $f_\lambda^0$ are Lipschitz. Similarly, $f_{\widehat{\lambda}}^\infty$ and $f_{\widehat{\lambda}}^0$ are Lipschitz.

Finally, to show that (A5) is bounded, and provide the upper bound in Theorem \ref{thm:inf-norm-rate-conv}, we need to show that $TV(f_{\widehat{\lambda}}^i, f_\lambda^i) \overset{p}{\to} 0$. We have
\begin{align}
|f_{\widehat{\lambda}}^\infty(s) - f_\lambda^\infty(s)| &\leq \left| \frac{1}{\widehat{\mu}} - \frac{1}{\mu} \right| \left| \frac{1}{s} \exp \left\lbrace -\frac{1}{2} \left( \frac{\log s}{\widehat{\mu}} + \frac{\widehat{\mu}}{2} \right) \right\rbrace \right| \\ 
& + \frac{1}{\mu} \left|  \frac{1}{s} \exp \left\lbrace -\frac{1}{2} \left( \frac{\log s}{\widehat{\mu}} + \frac{\widehat{\mu}}{2} \right)^2 \right\rbrace - \frac{1}{s}  \exp \left\lbrace -\frac{1}{2} \left( \frac{\log s}{\mu} + \frac{\mu}{2} \right)^2 \right\rbrace \right|.
\end{align}
First, 
\begin{align}
\int \limits_0^\infty \left| \frac{1}{\widehat{\mu}} - \frac{1}{\mu} \right| \left| \frac{1}{s} \exp \left\lbrace -\frac{1}{2} \left( \frac{\log s}{\widehat{\mu}} + \frac{\widehat{\mu}}{2} \right) \right\rbrace \right| ds = O_P\left( \left| \frac{1}{\widehat{\mu}} - \frac{1}{\mu} \right| \right) = O_P(|\widehat{\mu} - \mu|).
\end{align}
Next,
\begin{align}
\begin{split}
\left|  \frac{1}{s} \exp \left\lbrace -\frac{1}{2} \left( \frac{\log s}{\widehat{\mu}} + \frac{\widehat{\mu}}{2} \right)^2 \right\rbrace - \frac{1}{s}  \exp \left\lbrace - \frac{1}{2} \left( \frac{\log s}{\mu} + \frac{\mu}{2} \right)^2 \right\rbrace \right| &= |\widehat{\mu} - \mu| \left| \frac{1}{\mu^3 s} e^{-0.125(\mu^2 + 2 \log s)^2 / \mu^2} (\log^2(s) - 0.25\mu^4) \right| \\
& + O_P(|\widehat{\mu} - \mu|^2),
\end{split}
\end{align}
and 
\begin{align}
\int \limits_0^\infty |\widehat{\mu} - \mu| \left| \frac{1}{\mu^3 s} e^{-0.125(\mu^2 + 2 \log s)^2 / \mu^2} (\log^2(s) - 0.25\mu^4) \right| ds = O_P(|\widehat{\mu} - \mu|).
\end{align}

Therefore,
\begin{align}
TV(f_{\widehat{\lambda}}^i, f_\lambda^i) = \int \limits_0^\infty |f_{\widehat{\lambda}}^\infty(s) - f_\lambda^\infty(s)| ds \leq O_P(|\widehat{\mu} - \mu|).
\end{align}

\subsection{Details for Example \ref{ex:lda-label-shift}}
\label{sec:ex-2-details}

Under both pre- and post-change distributions, $X|Y = y \sim N(\bm{\mu_y}, \bm{\Sigma})$. Let $p(x) = \mathbb{P}_\infty(Y = 1 | X = x)$, and let $\mathcal{A}$ be the LDA classifier with predicted probabilities $\mathcal{A}(x) = \widehat{\mathbb{P}}_\infty(Y = 1 | X = x)$, given by \eqref{eq:lda-score}. Then, $\lambda(X)$ and $\widehat{\lambda}(X)$ are linear transformations of $p(X)$ and $\mathcal{A}(X)$ respectively, by \eqref{eq:scoreratio} and \eqref{eq:est-lr}. Therefore, to check the assumptions it is sufficient to consider $f_{p|Y=y}$ and $f_{\mathcal{A}|Y=y}$, the conditional densities of $p(X)|Y=y$ and $\mathcal{A}(X)|Y=y$.

First, we have 
\begin{align}
\begin{split}
f_{p|Y=y}(s) &= \frac{1}{s - s^2} \phi \left( \frac{ \log \left( \frac{s(1 - \pi_\infty)}{(1-s)\pi_\infty}\right) + \frac{1}{2} \left( \bm{\mu_1}^T \bm{\Sigma}^{-1} \bm{\mu_1} - \bm{\mu_0}^T \bm{\Sigma}^{-1} \bm{\mu_0} \right) - \bm{\mu_y}^T \bm{\Sigma}^{-1} (\bm{\mu_1} - \bm{\mu_0}) }{ \sqrt{ (\bm{\mu_1} - \bm{\mu_0})^T \bm{\Sigma}^{-1} (\bm{\mu_1} - \bm{\mu_0}) }  } \right) \\
f_{\mathcal{A}|Y=y}(s) &= \frac{1}{s - s^2} \phi \left( \frac{ \log \left( \frac{s(1 - \pi_\infty)}{(1-s)\pi_\infty}\right) + \frac{1}{2} \left( \widehat{\bm{\mu_1}}^T \widehat{\bm{\Sigma}}^{-1} \widehat{\bm{\mu_1}} - \widehat{\bm{\mu_0}}^T \widehat{\bm{\Sigma}}^{-1} \widehat{\bm{\mu_0}} \right) - \bm{\mu_y}^T \widehat{\bm{\Sigma}}^{-1} (\widehat{\bm{\mu_1}} - \widehat{\bm{\mu_0}}) }{ \sqrt{ (\widehat{\bm{\mu_1}} - \widehat{\bm{\mu_0}})^T \widehat{\bm{\Sigma}}^{-1} \bm{\Sigma} \widehat{\bm{\Sigma}}^{-1} (\widehat{\bm{\mu_1}} - \widehat{\bm{\mu_0}}) }  } \right),
\end{split}
\end{align}
where $\phi$ is the standard normal density. For simplicity, we write this as
\begin{align}
\begin{split}
f_{p|Y=y}(s) &= \frac{1}{s - s^2} \phi \left( \frac{\log \left( \frac{s(1 - \pi_\infty)}{(1-s)\pi_\infty}\right) + a}{b} \right) \\
f_{\mathcal{A}|Y=y}(s) &= \frac{1}{s - s^2} \phi \left( \frac{\log \left( \frac{s(1 - \pi_\infty)}{(1-s)\pi_\infty}\right) + \widehat{a}}{\widehat{b}} \right).
\end{split} 
\end{align}
Now, we show that these densities are Lipschitz. For $p(X)|Y=y$, we have
\begin{align}
\frac{d}{ds} f_{p|Y=y} = \frac{ \left(-a + b^2(2s - 1) - \log \left( \frac{s(1 - \pi_\infty)}{(1-s)\pi_\infty}\right) \right)  }{b^2(1 - s^2) s^2},
\end{align}
which is bounded. Similarly, $\frac{d}{ds} f_{\mathcal{A}|Y=y}$ is bounded. Therefore, $f_{p|Y=y}$ and $f_{\mathcal{A}|Y=y}$ are Lipschitz, so $f_\lambda^i$ and $f_{\widehat{\lambda}}^i$ are all Lipschitz, which satisfies (A3).

Next, we want to show convergence in total variation distance. Note that $|f_\lambda^i - f_{\widehat{\lambda}}^i| = O_P(|f_{p|Y=1} - f_{\mathcal{A}|Y=1}|) + O_P(|f_{p|Y=0} - f_{\mathcal{A}|Y=0}|)$, so it suffices to show that $TV(f_{p|Y=y}, f_{\mathcal{A}|Y=y})$ converges. Using $a, \widehat{a}, b, \widehat{b}$ from above, we get
\begin{align}
TV(f_{p|Y=y}, f_{\mathcal{A}|Y=y}) = \int |f_{\mathcal{A}|Y=y}(s) - f_{p|Y=y}(s)| ds = O_P(|\widehat{a} - a|) + O_P(|\widehat{b} - b|).
\end{align}
The right hand side converges by strong consistency of $\widehat{\bm{\mu_y}}$, $y \in \{0,1\}$, and $\widehat{\bm{\Sigma}}^{-1}$, satisfying (A5).

\end{document}